\documentclass[12pt, a4paper]{article}

\usepackage{bm}
\usepackage{amssymb}
\usepackage{amsmath}
\usepackage{amsthm}
\usepackage{stmaryrd}
\usepackage{verbatim}
\usepackage{textcomp}
\usepackage[hyphens]{url}
\usepackage[breaklinks]{hyperref} 
\usepackage{mathtools}
\usepackage[toc,page]{appendix}

\newcommand{\set}[2]{\{\,#1 \mid #2\,\}}

\newtheorem{defn}{Definition}
\newtheorem{thm}{Theorem}
\newtheorem{oldtheorem}{Theorem}

\newtheorem{lm}{Lemma}
\newtheorem{ex}{Example}

\newtheorem{conjecture}{Conjecture}
\newtheorem{cor}{Corollary}
\newtheorem{rem}{Remark}

\DeclareMathOperator{\poly}{poly}

\DeclareMathOperator{\rat}{rat}

\DeclareMathOperator{\Dual}{asSeries}

\DeclareMathOperator{\fin}{end}
\DeclareMathOperator{\final}{final}

\DeclareMathOperator{\Left}{Left}
\DeclareMathOperator{\Right}{Right}
\DeclareMathOperator{\Center}{Center}

\newcommand{\N}{\mathbb{N}}

\newcommand{\Q}{\mathbb{Q}}

\newcommand{\F}{\mathbb{F}_2}

\newcommand{\A}{\mathcal{A}}
\newcommand{\B}{\mathcal{B}}
\newcommand{\C}{\mathcal{C}}

\begin{document}

\title{Bounded Languages Described by GF(2)-grammars}
\author{Vladislav Makarov\thanks{This work was presented at the DLT 2021 conference held
	in Porto, Portugal on 16--20 August 2021, and its shorter
	version appeared in the conference proceedings:
	N. Moreira, R. Reis (Eds.), \emph{Developments in Language Theory},
	LNCS 12811, pp. 279--290.} \\ Saint Petersburg State University}

\maketitle

\begin{abstract}
GF(2)-grammars are a (somewhat) recently introduced grammar family that have some unusual 
algebraic properties and are closely connected to the family of unambiguous context-free grammars.
By using the method of formal power series, we establish strong conditions
that are necessary for a bounded language (a language is called bounded if it is a subset of $w_1^* w_2^* \ldots w_k^*$ for some positive integer $k$ and some strings $w_1$, $w_2$, $\ldots$, $w_k$) to be described by a GF(2)-grammar. 
By further applying the established results, we settle the long-standing open question of
proving the inherent ambiguity of the language $\set{a^n b^m c^\ell}{n \neq m \text{ or } m \neq \ell}$,
as well as give a new, purely algebraic, proof of the inherent ambiguity of the language $\set{a^n b^m c^\ell}{n = m \text{ or } m = \ell}$.


\end{abstract}

\sloppy

\section{Introduction}\label{intro}

GF(2)-grammars, recently introduced by Bakinova et al.~\cite{gf2},
and further studied by Makarov and Okhotin~\cite{grammars-gf2},
are a variant of ordinary context-free grammars (or just \emph{ordinary grammars}, 
as I will call them later in the text),
in which the disjunction is replaced by the exclusive OR,
whereas the classical concatenation
is replaced by a new operation called GF(2)-concatenation:
$K \odot L$ is the set of all strings
with an odd number of partitions
into a concatenation of a string in $K$ and a string in $L$.

There are several reasons for studying GF(2)-grammars. Firstly, they are a class of grammars
with better algebraic properties, compared to ordinary grammars and similar grammar families,
because the underlying boolean semiring logic is replaced by the logic of the field with two elements.  As we will
see later in the paper, that makes GF(2)-grammars lend themselves very well to algebraic manipulations.

Secondly, GF(2)-grammars provide a new way of looking at unambiguous grammars.
For example, instead of proving that some language is inherently ambiguous, one can prove that no GF(2)-grammar describes it.
While the latter condition is, strictly speaking, stronger, it may turn out to be easier to prove, 
because the family of languages defined by GF(2)-grammars has good algebraic properties
and is closed under symmetric difference.

Finally, GF(2)-grammars generalize the notion of parity nondeterminism to grammars.
Recall that the most common types of nondeterminism that are considered in complexity theory are classical
nondeterminism, which corresponds to the \emph{existence} of an accepting computation,
unambiguous nondeterminism, which corresponds to the existence of a \emph{unique} accepting computation
and parity nondeterminism, which corresponds to the \emph{number of accepting computations being odd}.

In a similar way,
classical and parity nondeterminism can be seen as two different generalisations of unambiguous nondeterminism: if the number 
of accepting computations is in the set $\{0,1\}$, then it is positive (classical case) if and only if it is odd (parity case); the same is not true for larger numbers, of course.

The main result of this paper is Theorem~\ref{Main_many}, which establishes a 
strong necessary conditions on subsets of $a_1^* a_2^* \cdots a_k^*$ that are described by GF(2)-grammars. Theorem~\ref{Main_abc}, a special case of Theorem~\ref{Main_many}, 
implies that there are no GF(2)-grammars for the languages 
$L_1 \coloneqq \set{a^n b^m c^\ell}{n = m \text{ or } m = \ell}$ and $L_2 \coloneqq \set{a^n b^m c^\ell}{n \neq m \text{ or } m \neq \ell}$.

As a consequence, both languages are inherently ambiguous. For $L_1$, all previously known arguments establishing its inherent ambiguity
were combinatorial, mainly based on Ogden's lemma. 

Proving the inherent ambiguity of $L_2$ was a long-standing open question due to Autebert et al.~\cite[p. 375]{autebert}. 
There is an interesting detail here: back in 1966, Ginsburg and Ullian fully characterized bounded
languages described by unambiguous grammars in terms of semi-linear sets~\cite[Theorems 5.1 and 6.1]{unamb-class}. However, most natural ways to apply this characterization suffer
from the same limitation: they mainly rely on strings that \emph{are not} in the language and much less on the strings
that \emph{are}. Hence, ``dense'' languages like $L_2$ leave them with almost nothing to work with. Moreover, $L_2$ has an algebraic generating function, meaning that a naive application of analytic methods cannot tackle it either. In fact, Flajolet~\cite{flajolet}, in his seminal work on analytic
methods for proving grammar ambiguity, refers to the inherent ambiguity of $L_2$ as to a question that is still open (see page 286).

Also, there are some other important results proved in this paper, both concerned with the special case of the subsets of $a^* b^*$.
Firstly, it is shown that there is a general family of subsets of $a^* b^*$ that can assuredly be described by a GF(2)-grammar. The
description of the said family is very similar to the upper bound on the describable subsets. Secondly, an example of a describable
subset of $a^* b^*$ is given that \emph{is not} a part of the given family, refuting the natural hypothesis that the obtained family
is in fact exactly the family of the subsets of $a^* b^*$ that can be described by GF(2)-grammars.

Finally, it is shown that the results on the subsets of $a_1^* a_2^* \ldots a_k^*$ (the so-called \emph{letter-bounded} languages)
can be extended to the general case of \emph{bounded} languages. Bounded languages are defined in the following way: a language $L \subseteq \Sigma^*$ is called bounded, if there exists a positive integer $k$ and $k$ non-empty strings $w_1$, $w_2$, $\ldots$, $w_k$ over the alphabet $\Sigma$, such that $L$ is a subset of $w_1^* w_2^* \ldots w_k^*$. Hence, the letter-bounded languages are a special case
of the general bounded languages; specifically, the case when each $w_i$ is actually a single character (and all $w_i$ are pairwise different).

It would be intuitive to expect that the structure of the general bounded languages described by GF(2)-grammars is much more complicated then the structure of the letter-bounded languages described by GF(2)-grammars. After all, the definition of a bounded language does not deal in any way with the
possibility of a string from $L$ having multiple different representations of the form $w_1^{\ell_1} w_2^{\ell_2} \ldots w_k^{\ell_k}$, where all $\ell_i$ are nonnegative. On the other hand, GF(2)-grammars are highly sensitive to ambiguity issues. However, somewhat surprisingly, Theorem~\ref{full-bounded} shows that the case of the general bounded languages is not much harder than the case of letter-bounded languages.

\section{Basics}\label{section_basic}

The proofs that we will see later make heavy use of algebraic methods.
For the algebraic parts, the exposition strives to be as elementary and self-contained as possible. Hence, I will prove a lot of lemmas that are by no way original and may be considered trivial by someone
with good knowledge of commutative algebra. This is the intended effect; if you consider something
to be trivial, you can skip reading the proof. If, on the other hand, you have some basic knowledge
of algebra, but still find some of the parts to be unclear, you may contact me and I will try to 
find a better wording. The intended ``theoretical minimum'' is being at least somewhat
familiar with concepts of polynomials, rational functions and formal power series.

Let us recall the definition and the basic properties of GF(2)-grammars first.
This section is completely based on already published work: the original paper about GF(2)-operations by Bakinova et al.~\cite{gf2} and the paper about basic properties of GF(2)-grammars by Makarov and Okhotin~\cite{grammars-gf2}. Hence, all the proofs are omitted; for proofs and 
more thorough commentary on definitions refer to the aforementioned papers. If you are already familiar with both of them, you may skip straight to the next section.
 
GF(2)-grammars are built upon GF(2)-operations~\cite{gf2}: symmetric difference and a new
operation called GF(2)-concatenation:
 \begin{align*}
	K \odot L
		= 
	\set{w}{\text{the number of partitions } w=uv,
		\text{ with } u \in K
		\text{ and } v \in L, \text{ is odd}}
\end{align*}
 
Syntactically, GF(2)-grammars do not differ from ordinary grammars.
However, in the right-hand sides of the rules, 
the normal concatenation is replaced with GF(2)-concatenation, whereas multiple
rules for the same nonterminal correspond to the symmetric difference of given conditions, instead of their disjunction.

\begin{defn}[\cite{gf2}]
A GF(2)-grammar is a quadruple $G=(\Sigma, N, R, S)$,
where:
\begin{itemize}
\item
	$\Sigma$ is the alphabet of the language;
\item
	$N$ is the set of nonterminal symbols;
\item
	every rule in $R$ is of the form
	$A \to X_1 \odot \ldots \odot X_\ell$,
	with $\ell \geqslant 0$ and $X_1, \ldots X_\ell \in \Sigma \cup N$,
	which represents all strings that have an odd number of partitions
	into $w_1 \ldots w_\ell$, with each $w_i$ representable as $X_i$;
\item
	$S \in N$ is the initial symbol.
\end{itemize}
The grammar must satisfy the following condition.
Let $\widehat{G}=(\Sigma, N, \widehat{R}, S)$ be the corresponding ordinary grammar,
with
$\widehat{R}=\set{A \to X_1 \ldots X_\ell}{A \to X_1 \odot \ldots \odot X_\ell \in R}$.
It is assumed that, for every string $w \in \Sigma^*$,
the number of parse trees of $w$ in $\widehat{G}$ is finite;
if this is not the case, then $G$ is considered ill-formed.

Then, for each $A \in N$,
the language $L_G(A)$ is defined as the set of all strings 
with an odd number of parse trees as $A$ in $\widehat{G}$.
\end{defn}

\begin{oldtheorem}[\cite{gf2}]\label{gf2_grammar_by_language_equations_theorem}
Let $G=(\Sigma, N, R, S)$ be a GF(2)-grammar.
Then the substitution $A=L_G(A)$ for all $A \in N$
is a solution of the following system of language equations.
\begin{align*}
	A
		=
	\bigtriangleup_{A \to X_1 \odot \ldots \odot X_\ell \in R}
	X_1 \odot \ldots \odot X_\ell
	&& (A \in N)
\end{align*}
\end{oldtheorem}

Multiple rules for the same nonterminal symbol
can be denoted by separating the alternatives with the ``sum modulo two'' symbol ($\oplus$),
as in the following example.

\begin{ex}[\cite{gf2}]
The following GF(2)-linear grammar defines the language
$\set{a^\ell b^m c^n}{\ell=m \text{ or } m=n, \text{ but not both}}$.
\begin{align*}
	S &\to A \oplus C \\
	A &\to a A \oplus B \\
	B &\to b B c \oplus \epsilon \\
	C &\to C c \oplus D \\
	D &\to a D b \oplus \epsilon
\end{align*}
Indeed, each string $a^\ell b^m c^n$ with $\ell=m$ or with $m=n$ has a parse tree,
and if both equalities hold, then there are accordingly two parse trees,
which cancel each other.
\end{ex}

\begin{ex}[\cite{gf2}]\label{example_2_to_n}
The following grammar describes the language $\set{a^{2^n}}{n \geqslant 0}$.
\begin{equation*}
	S \to (S \odot S) \oplus a
\end{equation*}
The main idea behind this grammar
is that the GF(2)-square $S \odot S$ over a unary alphabet
doubles the length of each string:
$L \odot L=\set{a^{2\ell}}{a^\ell \in L}$.
The grammar iterates this doubling to produce all powers of two.
\end{ex}

As the previous example illustrates, GF(2)-grammars can describe non-regular
unary languages, unlike ordinary grammars. 
We will need the classification of unary languages describable
by GF(2)-grammars in the following Sections.

\begin{defn}
A set of nonnegative integers $S \subseteq \N_0$
is called \emph{$q$-automatic}~\cite{automatic-sequences},
if there is a finite automaton
over the alphabet $\Sigma_q=\{0, 1, \ldots, q-1\}$
recognizing base-$q$ representations of these numbers.
\end{defn}

Let $\mathbb{F}_q [t]$
be the ring of polynomials over the $q$-element field GF($q$),
and let $\mathbb{F}_q [[ t ]]$
denote the ring  
of formal power series over the same field.

\begin{defn}
A formal power series $f \in \mathbb{F}_q [[ t ]]$ is said to be algebraic,
if there exists a non-zero polynomial $P$ with coefficients from $\mathbb{F}_q[t]$,
such that $P(f) = 0$.
\end{defn}

\begin{oldtheorem}[Christol's theorem for GF(2)~\cite{christol}]
A formal power series
$\sum_{n = 0}^{\infty} f_n t^n \in \mathbb{F}_2 [[t]]$
is algebraic if and only if
the set $\set{n \in \N_0}{f_n = 1}$ is 2-automatic.
\end{oldtheorem}

\begin{oldtheorem} [Unary languages described by GF(2)-grammars~\cite{grammars-gf2}] 
\label{unary-alphabet} For a unary alphabet, the class of all $2$-automatic languages
coincides with the class of all languages described by GF(2)-grammars.
\end{oldtheorem}

\section{Subsets of \texorpdfstring{$a^* b^*$}{a* b*}}\label{section_ab}

Suppose that some GF(2)-grammar over an alphabet $\Sigma = \{a, b\}$ 
generates a language that is a subset of $a^* b^*$.
How does the resulting language look like?  

It will prove convenient to associate subsets of $a^* b^*$ 
with (commutative) formal power series in two variables $a$ and $b$
over the field $\F$. This correspondence is similar to the correspondence 
between languages over a unary alphabet with GF(2)-operations 
($\odot, \triangle$) and formal power series of one variable 
with multiplication and addition~\cite{grammars-gf2}.

Formally speaking, for every set $S \subseteq \N_0^2$, 
the language $\set{a^n b^m}{(n, m) \in S} \subseteq a^* b^*$
corresponds to the formal power series $\sum_{(n, m) \in S} a^n b^m$
in variables $a$ and $b$. 
Let us denote this correspondence by
$\Dual \colon 2^{a^* b^*} \to \F[[a, b]]$. 
Then, $\Dual (L \triangle K) = \Dual (L) + \Dual (K)$, so the symmetric difference
of languages corresponds to the addition of power series.

On the other hand, multiplication of formal power series does not \emph{always} correspond 
to the GF(2)-concatenation of languages. 
Indeed, GF(2)-concatenation of subsets of $a^* b^*$ 
does not have to be a subset of $a^* b^*$. However, the correspondence does hold
in the following important special case.

\begin{lm}\label{comm-concat} If $K \subseteq a^*$ and $L \subseteq a^* b^*$, then 
$\Dual(K \odot L) = \Dual(K) \cdot \Dual(L)$. The same conclusion holds when 
$K \subseteq a^* b^*$ and $L \subseteq b^*$.
\end{lm}
\begin{proof}[Sketch of the proof] Follows from definitions.
\end{proof}

Denote the set of all algebraic power series from $\F[[a]]$ by $\A$.
By Christol's theorem~\cite{christol}, $\A$ corresponds to the set of all 2-automatic languages over $\{a\}$.
Similarly, denote the set of all algebraic power series from $\F[[b]]$ by $\B$. 

Recall that the $\F[a, b]$ denotes the set of all polynomials in variables
$a$ and $b$ and $\F(a, b)$ denotes the set of all rational functions in variables $a$ and $b$.
It should be mentioned that $\F[a, b]$ is a subset of $\F[[a, b]]$, but $\F(a, b)$ is not. Indeed, $\frac{1}{a} \in \F(a, b)$, but not in $\F[[a, b]]$.  The following statement is true: $\F(a, b) \subseteq \F((a, b))$, where $\F((a, b))$
denotes the set of all \emph{Laurent series} in variables $a$ and $b$. 
Laurent series are defined as the fractions of formal
power series with equality, addition and multiplication defined in the usual way.

We will allow Laurent series to appear in the intermediate results, because the 
intermediate calculations require division, and formal power series are not closed under
division. However, there are no Laurent series (unless they are valid formal power series as well) 
in the statements of the main theorems,  because they do not correspond to valid languages.

From now on, there are two possible roads this proof can take: the original argument that
is more elementary, but requires lengthy manipulations with what I called \emph{algebraic
expressions}, and a more abstract, but much simpler approach suggested by an anonymous
reviewer from MFCS 2020 conference, relying on well-known properties of rings and field extensions.
The main body of the paper follows the latter approach. The former approach can be found in 
the appendix.

\begin{defn}\label{ab-def}
Denote by $R_{a, b}$ the set of all Laurent series that can be represented
as $\frac{\sum_{i = 1}^n A_i B_i}{p}$, where $n$ is a nonnegative integer, $A_i \in \A$
and $B_i \in \B$ for all $i$ from $1$ to $n$, and $p \in \F[a, b]$ is a non-zero polynomial.
\end{defn}
It is not hard to see that $R_{a, b}$ is a commutative ring. However (and we will use it later a lot), a even stronger statement
is true:
\begin{lm}\label{lemma_rab} $R_{a, b}$ is a field. 
\end{lm}
\begin{proof} $R_{a,b}$ is the result of adjoining the elements of $\A \cup \B$, which are all algebraic over $\F(a,b)$, to $\F(a,b)$. It is known that the result of adjoining
an arbitrary set of algebraic elements to a field is a larger field.
\end{proof}

\subsection{The Main Result for Subsets of \texorpdfstring{$a^* b^*$}{a*b*}}

Let us establish our main result about subsets of $a^* b^*$.

\begin{thm}\label{Main_ab} Assume that a language $K \subseteq a^*  b^*$ 
is described by a GF(2)-grammar.
Then, the corresponding power series $\Dual(K)$ is in the set $R_{a, b}$.
\end{thm}

\begin{proof} Without loss of generality, the GF(2)-grammar 
that describes $K$ is in the Chomsky normal form~\cite[Theorem 5]{gf2}.  
Moreover, we can assume that $K$ does not contain the empty string.

The language $a^* b^*$ is accepted by the following DFA $M$:
$M$ has two states $q_a$ and $q_b$, both accepting, and its transition function is
$\delta(q_a, a) = q_a, \delta(q_a, b) = q_b, \delta(q_b, b) = q_b$.

Let us formally intersect the GF(2)-grammar $G$ with a regular language 
$a^* b^*$, recognized by the automaton $M$ 
(the construction of the intersection of an ordinary grammar with a regular expression by Bar-Hillel 
et al.~\cite{BarhillelPerlesShamir} can be easily adapted to the case of GF(2)-grammars~\cite [Section 6]{grammars-gf2}).
The language described by the GF(2)-grammar will not change, 
because it was already a subset of $a^* b^*$ before.

The grammar itself changes considerably, however. Every nonterminal
$C$ of the original GF(2)-grammar splits into three nonterminals: $C_{a \to a}$, $C_{a \to b}$.
and $C_{b \to b}$. These nonterminals will satisfy the following conditions: 
$L(C_{a \to a}) = L(C) \cap a^*$, $L(C_{b \to b}) = L(C) \cap b^*$ and
$L(C_{a \to b}) = L(C) \cap (a^* b^+)$.
Also, a new starting nonterminal $S'$ appears.

Moreover, every ``normal'' rule $C \to DE$ splits into four rules:
$C_{a \to a} \to D_{a \to a} E_{a \to a}$, $C_{a \to b} \to D_{a \to a} E_{a \to b}$, 
$C_{a \to b} \to D_{a \to b} E_{b \to b}$ and $C_{b \to b} \to D_{b \to b} E_{b \to b}$.

The following happens with ``final'' rules: $C \to b$ turns into two rules $C_{a \to b} \to b$ and $C_{b \to b} \to b$, and $C \to a$ turns into one rule $C_{a \to a} \to a$. 
Finally, two more rules appear: $S' \to S_{a \to a}$ and $S' \to S_{a \to b}$.

For every nonterminal $C$ of the original GF(2)-grammar, the languages
$L(C_{a \to a})$ and $L(C_{b \to b})$ are $2$-automatic languages over unary
alphabets $\{a\}$ and $\{b\}$ respectively. 
Indeed, every parse tree of $C_{a \to a}$ contains only nonterminals of type $a \to a$. 
Therefore, only character $a$ can occur as a \emph{terminal} in a parse tree of $C_{a \to a}$. 
So, $L(C_{a \to a})$ is described by some GF(2)-grammar over an alphabet $\{a\}$, 
and is therefore 2-automatic. Similarly for $C_{b \to b}$.
 
By Theorem~\ref{gf2_grammar_by_language_equations_theorem},
the languages $L(C_{a \to b})$ for each nonterminal $C_{a \to b}$ of the new grammar 
satisfy the following system of language equations (System~\eqref{LangSystem}).

Here, for each nonterminal $C$, the summation is over all rules $C \to DE$ of the original GF(2)-grammar. Also, $\fin(C_{a \to b})$ is either $\{b\}$ or $\varnothing$, depending on whether or not there is a rule $C_{a \to b} \to b$ in the new GF(2)-grammar.
\begin{equation}\label{LangSystem} L(C_{a \to b}) = \fin(C_{a \to b}) \oplus \bigoplus_{(C \to DE) \in R} 
(L(D_{a \to a}) \odot L(E_{a \to b})) \oplus (L(D_{a \to b}) \odot L(E_{b \to b}))
\end{equation}

It is easy to see that all GF(2)-concatenations in the right-hand sides satisfy the conditions
of Lemma~\ref{comm-concat}.
Denote $\Dual(L(C_{a \to b}))$ by $\Center(C)$, $\Dual(L(C_{a \to a}))$ by $\Left(C)$,
$\Dual(L(C_{b \to b}))$ by $\Right(C)$ and $\Dual(\fin(C_{a \to b}))$ by $\final(C)$ for brevity.
Then, the algebraic equivalent of System~\eqref{LangSystem} also holds:
\begin{equation}\label{SeriesSystem} \Center(C) = \final(C) + \sum\limits_{(C \to DE) \in R}
\Left(D) \Center(E) + \Center(D) \Right(E)
\end{equation}

Let us look at this system as a system of $\F[[a, b]]$-linear equations over variables $\Center(C) = \Dual(L(C_{a \to b}))$ for
every nonterminal $C$ of the original GF(2)-grammar.

We will consider $\final(C)$, $\Left(C)$ and $\Right(C)$ to be the coefficients of the system. 
While we do not know their \emph{exact} values, the following is known: $\final(C)$ is $0$ or $b$, $\Left(C) \in \A$ as a formal power series that
corresponds to a 2-automatic language over an alphabet $\{a\}$ and, similarly, $\Right(C) \in \B$.
That means that all cooeficients of the system lie in $\A \cup \B$ and, therefore, in $R_{a, b}$.
The latter is a field by Lemma~\ref{lemma_rab}.

Denote the number of nonterminals in the original GF(2)-grammar by $n$,
(so there are $n$ nonterminals of type $a \to b$ in the new GF(2)-grammar),
a column vector of values $\Center(C)$ by $x$ and a column vector of values 
$\final(C)$ in the same order by $f$. 
Let us fix the numeration of nonterminals $C$ of the old GF(2)-grammar.
After that, we can use them as the ``indices'' of rows and columns of matrices.

Let $I$ be an identity matrix of dimension $n \times n$, $A$ be a $n \times n$ matrix
with the sum of $\Left(D)$ over all rules $C \to DE$ of the original grammar standing on the intersection of $C$-th row and $E$-th column:
\begin{equation}\label{Adef}
 A_{C, E} \coloneqq \sum\limits_{(C \to DE) \in R} \Left(D) 
\end{equation}
Similarly, let $B$ be a $n \times n$ matrix with
\begin{equation}\label{Bdef}
 B_{C, D} \coloneqq \sum\limits_{(C \to DE) \in R} \Right(E)
\end{equation}

Then, the equation System~\eqref{SeriesSystem} can be rewritten as $x = f + (A + B) x$ in the
matrix form. In other words, $(A + B + I)x = f$. 
Consider a homomorphism $h \colon \F[[a, b]] \to \F$ that maps
power series to their constant terms (coefficients before $a^0 b^0$). Then, $h(\det(A + B + I))= \det(h(A + B + I)) = \det(h(A) + h(B) + h(I))$, where $h$ is extended to the $n \times n$
matrices with components from $\F[[a, b]]$ in the natural way (replace each component
of the matrix by its constant term). 

Because the new GF(2)-grammar for $K$ is also in Chomsky normal form, all languages $L(C_{a \to a})$ and $L(C_{b \to b})$ do not contain the empty string.
Therefore, all series $\Left(C) = \Dual(L(C_{a \to a}))$ and $\Right(C) = \Dual(L(C_{b \to b}))$ have zero constant terms. Hence,
$h(A) = h(B) = 0$, where by $0$ we mean a zero $n \times n$ matrix. On the other hand,
$h(I) =I$. Hence, $h(\det(A + B + I)) = \det(h(A) + h(B) + h(I)) = \det(I) = 1$.
Therefore, $\det(A + B + I) \neq 0$, because $h(0) = 0$.

Hence, the System~\eqref{SeriesSystem} has exactly one
solution within the field $\F((a, b))$~--- the actual values of $\Center(C)$. Moreover, we
know that all coefficients of the system lie in the field $R_{a, b} \subseteq \F((a, b))$. 
Therefore, all components of the unique solution also lie within the field $R_{a, b}$.
Hence, $\Dual(K) = \Dual(L(S')) = \Dual(L(S_{a \to a})) + \Dual(L(S_{a \to b}))$ also
lies in $R_{a, b}$.
\end{proof}
\begin{rem} Alternatively, one can prove the uniqueness of the solution to System~\eqref{SeriesSystem} by some kind of fixed-point argument. However, I
stick to proving that the determinant is non-zero, mainly because the proof of Theorem~\ref{Main_abc} still uses the existence of an inverse matrix regardless of how the uniqueness of the solution is established. 
\end{rem}

\subsection{Using Theorem~\ref{Main_ab}}

It is hard to use Theorem~\ref{Main_ab} directly. 
Hence, we will prove the following intermediate result:

\begin{thm}\label{ab_structure} Suppose that $L \subseteq a^* b^*$ is described by 
a GF(2)-grammar. Denote ``the coefficient'' of $\Dual(L)$ before $a^i$ by $\ell(i) \in \F[[b]]$,
in the sense that $\Dual(L) = \sum_{i=0}^{+\infty} a^i \ell(i)$.
Then, there exists a nonnegative integer $d$ and polynomials $p_0, p_1, \ldots, p_d \in \F[b]$,
such that $p_d \neq 0$ and  $\sum_{i=0}^d p_i \ell(n-i)$ assumes only a finite number of distinct values, when $n$ ranges over the set of all integers larger than $d$.
\end{thm}

\begin{ex} For example, suppose that $\Dual(L) = \dfrac{A_1 B_1 + A_2 B_2}{1 + ab}$, 
where $A_1, A_2 \in A$ and $B_1, B_2 \in \B$. Then $\Dual(L) (1 + ab) = A_1 B_1 + A_2 B_2$.
Denote the coefficient of $\Dual(L)$ before $a^n$ by $\ell(n)$. Then 
$\left( \sum_{n=0}^{+\infty} a^n \ell(n) \right) (1+ab) = A_1 B_1 + A_2 B_2$.
Coefficients of the left-hand side before $a^n$ are $b \cdot \ell(n-1) + \ell(n)$ for $n \geqslant 1$.
The corresponding coefficients of the right-hand side are always from the set $\{0, B_1, B_2, B_1 + B_2\}$.
Therefore, it is enough to choose $d = 1, p_0 = 1, p_1 = b$ in this case. 
\end{ex}

\begin{rem} Actually, $\sum_{i=0}^d p_i \ell(n-i)$ is
\emph{2-automatic sequence}~\cite{automatic-sequences} of elements of $\B$.
That is, only elements from $\B$ appear in this sequence, only finite number of them actually appear,
and every element appears on a 2-automatic set of positions.
We just will not need the result in the maximum possible strength here.
\end{rem}
\begin{proof}[Proof of Theorem~\ref{ab_structure}]
As we already know, $\Dual(L) \in R_{a, b}$, meaning that
$\Dual(L) = \left( \sum_{k = 1}^K A_k B_k \right) \left( \sum_{i=0}^{d} a^i p_i \right)$ for
some nonnegative integers $d$ and $K$, $A_i \in \A$, $B_i \in \B$ and $p_i \in \F[b]$. Moreover, we can choose
$d$ in such a way, that $p_d \neq 0$: not all $p_i$ are equal to zero, because otherwise the denominator
of the fraction would be equal to zero. Also, $\Dual(L) = \sum\limits_{j=0}^{+\infty} a^j \ell(j)$ by definition of 
$\ell (\cdot)$.

Therefore, 
$\left( \sum_{i=0}^d a^i p_i \right) \cdot \left( \sum_{j=0}^{+\infty} a^j \ell(j) \right) =
\sum_{k=1}^K A_k B_k$. The coefficients of the left-hand and the right-hand sides before $a^n$
are $\sum_{i=0}^{\min(n, d)} p_i \ell (n - i)$ and $\sum_{k \colon a^n \in A_k} B_k$ respectively. Here, the second sum is taken over all $k$ from $1$ to $K$, such that the coefficient
of $A_k$ before $a^n$ is one.
When $n \geqslant d$, the former of these two values is $\sum_{i=0}^{d} p_i \ell (n - i)$ and 
the latter always takes one of $2^K$ possible values. 
\end{proof}

Let us consider a simple application of Theorem~\ref{ab_structure} to get the hang of how
it can be used to prove something.

\begin{thm}\label{a2n_b2n}
The language $K = \set{a^{2^n} b^{2^n}}{n \in \N}$ is not described by a GF(2)-grammar.
\end{thm}

\begin{proof} By contradiction. Let us use Theorem~\ref{ab_structure} on the language $K$. The coefficient $\ell(n)$
of $\Dual(K)$ is $b^n$, if $n$ is a power of two and $0$ otherwise. In any case, it is divisible by $b^n$. On one hand, from
the conclusion of Theorem~\ref{ab_structure}, $\sum_{i = 0}^d p_i \ell(n - i)$ assumes only a finite number
of values for some integer $d$ and $p_0, p_1, \ldots, p_d \in \F[b]$, satisfying the condition $p_d \neq 0$.

On the other hand, the sum $\sum_{i = 0}^d p_i \ell(n - i)$ is divisible by $b^{n- d}$:
every summand contains a factor $\ell(n-i)$, which is divisible by $b^{n-i}$. Therefore,
$\sum_{i = 0}^{d} p_i \ell(n - i)$ is divisible by larger and larger powers of $b$ as $n$ grows.
Therefore, the only value that this sum can assume infinitely often is $0$: all other power series are not divisible
by arbitrarily large powers of $b$. Because the sum assumes only a finite number of values, $0$ is obtained for 
large enough $n$.

Therefore, some fixed linear combination of $\ell(n-d), \ell(n-d+1), \ldots, \ell(n)$ is equal to $0$ for large enough $n$.
However, non-zero values appear in the sequence $\ell(i)$ extremely rarely: the gaps between them grow larger and larger.
In particular, one can choose such $n \geqslant d$, that $\sum_{i=0}^d p_i \ell(n-i) = 0$, $\ell(n-d) \neq 0$, but
$\ell(n-d+1) = \ldots = \ell(n) = 0$.
This is impossible, because $p_d \neq 0$ and, therefore, there is exactly one non-zero summand in a zero sum: $p_d \ell (n - d)$.

To be exact, one can pick $n = d + 2^m$ for large enough $m$. 
\end{proof}

A more interesting application of the technique can be seen in the following theorem
(a weaker result was earlier obtained by Makarov and Okhotin~\cite[Theorem~12]{grammars-gf2}
using elementary methods).

\begin{thm}\label{an_bfn_increasing}
Suppose that $f \colon \N_0 \to \N_0$ is a strictly increasing function.
If a language $L_f = \set{a^n b^{f(n)}}{n \in \N_0}$ is generated by a GF(2)-grammar,
then the set $f(\N_0)$ is a finite union of arithmetic progressions.
\end{thm}

\begin{proof} 
Let us use Theorem~\ref{ab_structure} here.
The proof is structured in the following way. The first step is to prove that
polynomials $b^{f(n)}$ satisfy some linear reccurence that has rational functions
of $b$ as coefficients. The second step is to prove that $f(\N_0)$ is, indeed,
a finite union of arithmetic progressions. Intuitively, it is hard to imagine a
linear recurrence with all values looking like $b^{\text{something}}$, but without
strong regularity properties.

Let us use Theorem~\ref{ab_structure} on the language $L_f$. The coefficient of
$\Dual(L_f)$ before $a^n$ is $b^{f(n)}$. Therefore, $\sum_{i=0}^d p_i b^{f(n-i)}$
assumes only a finite number of values for some nonnegative integer $d$ and some 
$p_0, p_1, \ldots, p_d \in \F[b]$ satisfying the property $p_d \neq 0$. Notice
that this sum starts being divisible by arbitrarily large powers of $b$ when $n$ increases
(here we use the fact that $f$ is an increasing function). Because the sum assumes
only a finite number of values, it is equal to zero for large enough $n$.

Now, we want to prove that $\set{f(n)}{n \in \N_0}$ is a finite union of arithmetic
progressions. We have already established that $\sum_{i=0}^d p_i b^{f(n-i)} = 0$
for large enough $n$.

Let $j$ be the smallest index, such that $p_j \neq 0$: it exists, because $p_d \neq 0$.
Moreover, $j \neq d$, otherwise $0 = p_d b^{f(n-d)}$ for large enough $n$,
contradicting $p_d \neq 0$. Therefore, $p_j b^{f(n-j)} = \sum_{i=j+1}^d p_i b^{f(n-i)}$.
Let us rewrite the last statement in a slightly different way:
\begin{equation*}
b^{f(n-j)} = \sum\limits_{i=j+1}^d \dfrac{p_i}{p_j} b^{f(n-i)}.
\end{equation*}
Divide both sides of the last equation by $b^{f(n-d)}$. Then,
\begin{equation*}
b^{f(n-j) - f(n-d)} = \sum\limits_{i=j+1}^d \dfrac{p_i}{p_j} b^{f(n-i) - f(n-d)}.
\end{equation*}
Therefore, the difference $f(n-j) - f(n-d)$ depends only on differences 
$f(n-j-1) - f(n-d), \ldots, f(n-d+1)-f(n-d)$, but not on $f(n-d)$ itself.

Now, let us prove that the difference $f(n+1) - f(n)$ is bounded above
(it is always positive, because $f$ is increasing). Indeed, as we know, 
$\sum_{i=0}^d p_i b^{f(m-i)} = 0$ for large enough $m$.
Substitute $n \coloneqq m - d$ and $k \coloneqq d - i$, the result is 
$\sum_{k=0}^{d} p_{d-k} b^{f(n+k)} = 0$.
Because $f$ increases, all summands are divisible by $b^{f(n+1)}$, with a
possible exception of the first summand. Because the sum is equal to $0$,
the first summand should be divisible by $b^{f(n+1)}$ as well. 
It is not equal to $0$ (because $p_d \neq 0$ and $b^{f(n)} \neq 0$)
and its degree as a polynomial of $b$ is equal to $\deg p_d + \deg b^{f(n)} = \deg p_d + f(n)$.
The degree of a non-zero polynomial divisible by $b^{f(n+1)}$ is at least $f(n+1)$,
therefore $f(n+1) - f(n) \leqslant \deg p_d$.

Because the differences $f(n+1) - f(n)$ are bounded, then the differences
$f(n+k) - f(n) = (f(n+k)-f(n+k-1)) + (f(n+k-1) - f(n+k-2)) + \ldots + (f(n+1)-f(n))$
are bounded as well for all $k \leqslant d$. Therefore, 
the tuple of differences $(f(n-j-1) - f(n-d), \ldots, f(n-d+1)-f(n-d))$ assumes only a finite
set of possible values as $n$ goes towards infinity. As shown above, $f(n-j) - f(n-d)$
can be uniquely restored from such a tuple. Therefore, the tuple for $n+1$ can be
uniquely restored from the tuple for $n$:
indeed, it is enough to know the pairwise differences
between the elements of $\{f(n-j), f(n-j-1), \ldots, f(n-d)\}$, and
the current tuple along with the number $f(n-j) - f(n-d)$ provide this information.

Because there is only a finite number of such tuples, and each tuple determines the next,
they start ``going in circles'' at some moment. In particular, the differences $f(n-d+1) - f(n-d)$
start going in circles. This fact, along with the function $f$ being increasing, is enough to establish
that $\set{f(n)}{n \in \N_0}$ is a finite union of arithmetic progressions.
\end{proof}

\section{Subsets of \texorpdfstring{$a^* b^* c^*$}{a* b* c*}}\label{section_abc}


The language $\set{a^n b^n c^n}{n \geqslant 0}$ is, probably, the most famous 
example of a simple language that is not described by any ordinary grammar. 
It is reasonable to assume that it is not described by a GF(2)-grammar as well.
Let us prove that.

We will do more than that and will actually establish some property that
all subsets of $a^* b^* c^*$ that can be described by a GF(2)-grammar have, but 
$\set{a^n b^n c^n}{n \geqslant 0}$ does not. 
Most steps of the proof will be analogous to the two-letter case.

There is a natural one-to-one correspondence between subsets of $a^* b^* c^*$ 
and formal power series in variables $a, b$ and $c$ over field $\F$. Indeed, for every
set $S \subseteq \N_0^3$, we can identify the language $\set{a^n b^m c^k}{(n, m, k) \in S}
\subseteq a^* b^* c^*$ with the formal power series 
$\sum_{(n, m, k) \in S} a^n b^m c^k$. Denote this correspondence 
by $\Dual \colon 2^{a^* b^* c^*} \to \F[[a, b, c]]$. Then, $\Dual(L \triangle K)
= \Dual(L) + \Dual(K)$. In other words, the symmetric difference of languages
corresponds to the sum of formal power series.

Similarly to the Lemma~\ref{comm-concat}, $\Dual(K \odot L) = \Dual(K) \cdot \Dual(L)$
in the following important special cases: when $K$ is a subset of $a^*$, 
when $K$ is a subset of $a^* b^*$ and $L$ is a subset of $b^* c^*$, and, finally, when $L$
is a subset of $c^*$.	 Indeed, in each of these three cases, characters ``are in the correct order'':
if $u \in K$ and $v \in L$, then $uv \in a^* b^* c^*$.

However, we cannot insert character $b$ in the middle of the string: if $K$ is a subset 
of $b^*$ and $L$ is a subset of $a^* b^* c^*$, then $K \odot L$ does not even have to 
be a subset of $a^* b^* c^*$.

The ``work plan'' will remain the same as in the previous section: we will prove
that some algebraic structure is a field and then use linear algebra over said 
field. Let us get two possible questions out of the way first:

\begin{enumerate}
\item Why is it logical to expect that the language $\set{a^n b^n c^n}{n \geqslant 0}$
is not described by a GF(2)-grammar, but a similar language $\set{a^n b^n}{n \geqslant 0}$
is?
\item Why will the proof work out for $\set{a^n b^n c^n}{n \geqslant 0}$, but not
for a regular language $\set{(abc)^n}{n \geqslant 0}$, despite these languages
having the same ``commutative image''?
\end{enumerate}

They can be answered in the following way:

\begin{enumerate}
\item Simply speaking, the reason is the same as for the ordinary grammars. On a
intuitive level, both ordinary grammars and GF(2)-grammars permit a natural
way to ``capture'' the events that happen with any two letters in subsets of $a^* b^* c^*$,
but not all three letters at the same time. A rigourous result that corresponds to this 
intuitive limitation of ordinary grammars was proven by Ginsburg and Spanier~\cite[Theorem 2.1]{bounded-original}. Theorem~\ref{Main_abc} is an analogue for GF(2)-grammars.

\item This argument only implies that any proof that relies \emph{solely} on commutative
images is going to fail. The real proof is more subtle. For example, it will also use the fact
that $\set{a^n b^n c^n}{n \geqslant 0}$ is a subset of $a^* b^* c^*$.

While the proof uses commutative images, it uses them very carefully, always making sure
that the letters ``appear in the correct order''. In particular, we will never consider
GF(2)-concatenations $K \odot L$, where $K$ is a subset of $b^*$
and $L$ is an arbitrary subset of $a^* b^* c^*$, in the proof, 
because in this case $K \odot L$ is not a subset of $a^* b^* c^*$.

Avoiding this situation is impossible for language $\set{(abc)^n}{n \geqslant 0}$, because
in the string $abcabc$ from this language the letters ``appear in the wrong order''.
\end{enumerate}

Denote the set of algebraic power series in variable $c$ by $\C$. 

Similarly to Definition~\ref{ab-def}, define $R_{a, c} \subseteq \F((a, c))$ and 
$R_{b, c} \subseteq \F((b, c))$. 

Finally, denote by $R_{a, b, c}$ the set of all Laurent series that can be represented
as 
\begin{equation*}
\dfrac{\sum\limits_{i=1}^n A_i B_i C_i}{p_{a, b} \cdot p_{a, c} \cdot p_{b, c}},
\end{equation*} where $n$ is a nonnegative integer, $A_i \in \A$,
$B_i \in \B$, $C_i \in \C$ for all $i$ from $1$ to $n$, and $p_{a, b} \in \F[a, b]$, $p_{a, c} \in \F[a, c]$,
$p_{b, c} \in \F[b, c]$. 

\begin{lm}\label{lemma_abc} $R_{a, b, c}$ is a ring and a subset of $\F((a, b, c))$. 
Moreover, 
$R_{a, b}$, $R_{a, c}$ and $R_{b, c}$ are subsets of $R_{a, b, c}$.
\end{lm}
\begin{proof} It is easy to see that $R_{a, b, c}$ is closed under addition and multiplication.
Setting $C_1 = C_2 = \ldots = C_n = p_{a, c} = p_{b, c} = 1$ yields $R_{a, b} \subseteq R_{a, b, c}$. 
\end{proof}

\subsection{Main Result}

Unlike $R_{a, b}$, $R_{a, b, c}$ \emph{is not} a field (in fact, Subsection~\ref{section_anbncn} tells us that $(1 + abc)^{-1} \notin R_{a, b, c}$), so a bit more involved argument will be necessary for
the proof of the following theorem:

\begin{thm}\label{Main_abc} Suppose that $K \subseteq a^* b^* c^*$ is described by a GF(2)-grammar.
Then the corresponding formal power series $\Dual(K)$ is in the set $R_{a, b, c}$. 
\end{thm}

\begin{proof}[Proof outline] The proof is mostly the same as the proof of
Theorem~\ref{Main_ab}. Let us focus on the differences. As before, we can assume that $K$ does not contain the empty string.

In the same manner, we formally intersect our GF(2)-grammar in Chomsky's normal form with the language
$a^* b^* c^*$. Now, all nonterminals $C$ of the original GF(2)-grammar split
into \emph{six} nonterminals: $C_{a \to a}, C_{a \to b}, C_{a \to c}, C_{b \to b}, C_{b \to c},
C_{c \to c}$. However, their meanings stay the same. 
for example, $L(C_{a \to b}) = L(C) \cap (a^* b^+)$ and
$L(C_{a \to c}) = L(C) \cap (a^* b^* c^+)$.

However, only the ``central'' nonterminals $C_{a \to c}$ are important,
similarly to the nonterminals of the type $a \to b$ in the proof of Theorem~\ref{Main_ab}.
Why? Before, we had some a priori knowledge about the languages $L(C_{a \to a})$ and
$L(C_{b \to b})$ from Christol's theorem. But now, because of Theorem~\ref{Main_ab}, 
we have a priori knowledge about the
languages $L(C_{a \to b})$ and $L(C_{b \to c})$ as well, because they are subsets
of $a^* b^*$ and $b^* c^*$ respectively.

\begin{rem} In a sense, we used Theorem~\ref{unary-alphabet} as a stepping stone towards the proof
of Theorem~\ref{Main_ab}, and now we can use Theorem~\ref{Main_ab} as a stepping stone
towards the proof of Theorem~\ref{Main_abc}.
\end{rem}
Denote by $\fin(C)$ the language $\left( \bigoplus\limits_{(C \to DE) \in R} L(D_{a, b}) \odot L(E_{b, c})  \right) \oplus T_C$, where $T_C$ is either $\{c\}$ or $\varnothing$, depending
on whether or not there is a ``final'' rule $C_{a \to c} \to c$ in the new GF(2)-grammar.

This means that we again can
express the values $\Dual(L(C_{a \to c}))$ as a solution to a system of linear equations
with \emph{relatively simple} coefficients (denote $\Dual(L(C_{a \to c}))$ by
$\Center(C)$, $\Dual(L(C_{a \to a}))$ by $\Left(C)$, $\Dual(L(C_{c \to c}))$ by $\Right(C)$
and $\Dual(\fin(C))$ by $\final(C)$):
\begin{equation}\label{series_ab}
\Center(C) = \final(C) + \sum\limits_{(C \to DE) \in R} \Left(D) \Center(E) + \Center(D) \Right(E)
\end{equation}
Here, the summation is over all rules $C \to DE$ of the original GF(2)-grammar. Similarly to
the proof of Theorem~\ref{Main_ab}, this system can be rewritten as $(A + B + I)x = f$, 
where $x$ and $f$ are column-vectors of $\Center(C)$ and $\final(C)$ respectively, 
while the matrices $A$ and $B$ are defined as follows:
\begin{align*}\label{abc_matrices}
 A_{C, E} &\coloneqq \sum\limits_{C \to DE} \Left(D) \\ 
 B_{C, E} &\coloneqq \sum\limits_{C \to DE} \Right(E)
\end{align*}
Again, this system has a unique solution, because we can prove that $\det(A + B + I) \neq 0$. Moreover, said solution can be written down in the following way: $x = (A + B + I)^{-1} f$. 
Recall that $R_{a, c}$ is a field and all entries of $A + B + I$ lie in $R_{a, c}$. Hence, all
entries of $(A + B + I)^{-1}$ are elements of $R_{a, c}$ as well. All entries of $f$ are elements
of $R_{a, b, c}$. Therefore, all entries of $x$ lie in $R_{a, b, c}$ (here
we use that $R_{a, c}$ is a subring of $R_{a, b, c}$).
In particular, $\Center(S) \in R_{a, b, c}$. Therefore, $\Dual(K) = \Dual(L(S')) = \Dual(L(S_{a \to a}))
+ \Dual(L(S_{a \to b})) + \Center(S)$ also lies in the set.
\end{proof}

Consider the case of larger alphabets. Let $\A_i$ be the set of all algebraic formal power
series in variable $a_i$. Similarly to $R_{a, b, c}$, denote by $R_{a_1, a_2, \ldots, a_k}$ the set of all Laurent
series that can be represented as $\dfrac{\sum_{i = 1}^n A_{i, 1} A_{i, 2} \ldots A_{i, k}}{\prod\limits_{1 \leqslant i < j \leqslant n} p_{i, j}}$, for some $n \geqslant 0$, $p_{i, j} \in \F(a_i, a_j)$ and $A_{i, j} \in \A_j$.

\begin{thm}\label{Main_many} If a language $K \subseteq a_1^* a_2^* \ldots a_k^*$ is described
by a GF(2)-grammar, then the corresponding power series $\Dual(K)$ is in the set
$R_{a_1, a_2, \ldots, a_k}$. 
\end{thm}
\begin{proof}[Sketch of the proof]
Induction over $k$, the induction step is analogous to the way we used Theorem~\ref{Main_ab} in the proof of Theorem~\ref{Main_abc}.
\end{proof}

There is another thing to note here. There is a common pattern that I used to ensure that a GF(2)-grammar describing a subset of $a_1^* a_2^* \ldots a_k^*$ is in a very specific form: converting it to Chomsky normal form, removing the empty string from the language and
formally intersecting the GF(2)-grammar with a DFA describing the regular language $a_1^* a_2^* \ldots a_k^*$. Only one of the
performed operations (removing the empty string) may change the described language (and, even then, the change is negligible for almost any purposes), but the form of the GF(2)-grammar changes
considerably. In some contexts, mainly algorithmic ones, we are interested in the specific grammars used to describe the language. Here, however, we are only interested in the expressive power of GF(2)-grammars in general, without any consideration of the size of the GF(2)-grammar. 

One may look at the performed transformation in the following way: the transformation itself proves than any subset of $a_1^* a_2^* \ldots a_k^*$ that can be described by a GF(2)-grammar at all, can be described by a GF(2)-grammar of some special form; then, it is possible to just say that the GF(2)-grammar is in the necessary form without manually performing the transformation every time. Now, how to define the special form used? Well, one of the ways is just to say that the GF(2)-grammar is special form if it is the result of applying the transformation to some original GF(2)-grammar. However, such a definition is not the best, as the proofs of Theorems~\ref{Main_ab} and~\ref{Main_abc} do not actually use all properties of the obtained GF(2)-grammar; they only care about the types of nonterminals ``matching up'' in the every rule. Hence, it makes sense to give the following two definitions. 

\begin{defn} A GF(2)-grammar is in the \emph{prototype letter-bounded normal form}, if it satisfies the following properties:
\begin{enumerate}
\item There is one special starting nonterminal. Each of the other nonterminals is assigned some \emph{type} $a_i \to a_j$, where 
$i \leqslant j$. The nonterminal $D$ having type $a_i \to a_j$ means that $L(D)$ is a subset of 
$a_i^* a_{i+1}^* \ldots a_{j-2}^* a_{j-1}^* a_j^+$. In other words, only characters from $a_i$ to $a_j$ can appear in 
the strings from $L(D)$, they should appear in the correct order, and at least one copy of $a_j$ must appear.
\item For a starting nonterminal $A$ of the type $a_i \to a_j$ with $i \leqslant j$, there are only two possible ``types'' of rules: the ``final'' rules $A \to a_j$ and the 
``normal'' rules $A \to BC$, where $B$ and $C$ are nonterminals of the types $a_i \to a_m$ and $a_m \to a_j$, respectively, for some integer $m$ from $i$ to $j$.
\item For the starting nonterminal $S$ and for each $i$ from $1$ to $k$, there is at most one rule of the form $S \to A_k$, where $A_k$ is a non-starting nonterminal of the type $a_1 \to a_k$. There also may be at most one rule $S \to \varepsilon$. There are no other rules for the
nonterminal $S$.
\end{enumerate}
\end{defn}

An existence of a GF(2)-grammar in the prototype letter-bounded normal form for every subset of $a_1^* a_2^* \ldots a_k^*$
that can be described by any GF(2)-grammar immediately follows from the transformation (it may be necessary to add the rule $S \to \varepsilon$ back if it was previously deleted). Similarly, it is clear that every GF(2)-grammar in the prototype letter-bounded
normal form describes a subset of $a_1^* a_2^* \ldots a_k^*$.

The definition of the prototype letter-bounded normal form precisely captures the required properties of the transformed GF(2)-grammar.
However, it is sometimes better to work with a GF(2)-grammar that is in Chomsky normal form. While the prototype letter-bounded
normal form is \emph{almost} a special case of the letter-bounded normal form, the special starting ensures that it does not meet
the precise definition of Chomsky normal form. Hence, it makes sense to consider the following definition.

\begin{defn} A GF(2)-grammar is in the \emph{letter-bounded normal form}, it if satisfies the following properties.
\begin{enumerate}
\item Each of its nonterminals is assigned some \emph{type} $a_i \to a_j$, where $i \leqslant j$. The
nonterminal $D$ having type $a_i \to a_j$ means that $L(D)$ is a subset of $a_i^* a_{i+1}^* \ldots a_{j-2}^* a_{j-1}^* a_j^+$. In other words, 
only characters from $a_i$ to $a_j$ can appear in the strings from $L(D)$, they should appear in the correct order, and at least one copy of $a_j$ must appear.
\item The starting nonterminal is of the type $a_1 \to a_k$.
\item For a nonterminal $A$ of the type $a_i \to a_j$ with $i \leqslant j$, there are only two possible ``types'' of rules: the ``final'' rules $A \to a_j$ and the 
``normal'' rules $A \to BC$, where $B$ and $C$ are nonterminals of the types $a_i \to a_m$ and $a_m \to a_j$, respectively, for some integer $m$ from $i$ to $j$.
\end{enumerate}
\end{defn}

It is not hard to see that the letter-bounded normal form is a special case of Chomsky normal form. Moreover, a language described
by a GF(2)-grammar in the letter-bounded normal form is a subset of $a_1^* a_2^* \ldots a_{k-1}^* \cdot a_k^+$. In particular,
the described language cannot contain the string $\varepsilon$. Finally, the general construction still shows that each subset of
 $a_1^* a_2^* \ldots a_{k-1}^* \cdot a_k^+$ that can be described by a GF(2)-grammar at all, can be described by one in
the letter-bounded normal form: we just need to notice that the only purpose of adding a new starting nonterminal is to ensure
that all the ``disjoint'' components of the language $L \subseteq a_1^* a_2^* \ldots a_k^*$ that are defined by iterating over the 
last character of a string from $L$, are somehow in the described language: if the new starting nonterminal is $S'$ and the old starting nonterminal is $S$, then there is a rule $S' \to S_{a_1 \to a_k}$, describing $L \cap (a_1^* a_2^* \ldots a_{k-1}^* \cdot a_k^+)$, 
a rule $S \to S_{a_1 \to a_{k-1}}$, describing $L \cap (a_1^* a_2^* \ldots a_{k-2}^* \cdot a_{k-1}^+)$, et cetera. In fact, these rules are
exactly the rules from the definition of the prototype letter-bounded normal form. 

Again, the good property of both of these normal forms is that their syntax ensures that the described language exists and is a subset
of $a_1^* a_2^* \ldots a_k^*$ or a subset of $a_1^* a_2^* \ldots a_{k-1}^* \cdot a_k^+$. Indeed, it is easy to prove that each 
GF(2)-grammar in Chomsky normal form has only finitely many parse trees for each string. Therefore, even in the definition of the prototype
letter-bounded normal form, there are finitely many parse trees for each string from each nonterminal. Indeed, the number of parse
trees for $w$ from $S$ is the number of parse trees for one of its following nonterminals (specifically, the one that has the same type as the string $w$). So, in a sense, the definitions ``syntactically'' encode the fact the described language has the necessary properties.

Another thing to note is that the both definitions here actually do not define a GF(2)-grammar, but a GF(2)-grammar with some
additional information (the types of the nonterminals). Recovering the types of nonterminals seems to be difficult: not only could
there be several different ways to recover the types, but also the problem of recovering any valid assignment of the types seems
to be connected with the problem of deciding the emptiness of a GF(2)-grammar. Currently, it is unknown whether or not
the emptiness problem is decidable for GF(2)-grammars. Hence, it may be more proper to say that both these normal forms are not
GF(2)-grammars, but rather GF(2)-grammars with some additional baked-in information. 

\subsection{The Language \texorpdfstring{$\set{a^n b^n c^n}{n \geqslant 0}$}{an bn cn} and Its Relatives} 
\label{section_anbncn}

In this subsection, we will use our recently obtained knowledge to 
prove that there is no GF(2)-grammar for the language $\set{a^n b^n c^n}{n \geqslant 0}$.
It will almost immediately follow that the languages $\set{a^n b^m c^\ell}{n = m \texttt{ or } m = \ell}$ and $\set{a^n b^m c^\ell}{n \neq m \texttt{ or } m \neq \ell}$

Consider the formal power series $\Dual(\set{a^n b^n c^n}{n \geqslant 0}) = 
\sum_{n=0}^{+\infty} a^n b^n c^n$. Denote these series by $f$ for brevity. It is not
hard to see that $f = (1 + abc)^{-1}$. Indeed, $f \cdot (1 + abc) = 
\sum_{n = 0}^{+\infty} (a^n b^n c^n + a^{n+1} b^{n+1} c^{n+1}) = 1$,
because all summands except $a^0 b^0 c^0 = 1$ cancel out.

It sounds intuitive that $(1 + abc)^{-1}$ ``depends'' on $a$, $b$ and $c$ in a way
that the $R_{a, b, c}$ cannot capture;
series in $R_{a, b, c}$ should ``split'' nicely into functions that depend
only on two variables out of three. Now, let us establish that $f \notin R_{a, b, c}$ formally.

Indeed, suppose that it is not true. In other words,
\begin{equation}\label{f_definition} f = \dfrac{\sum\limits_{i=1}^n A_i B_i C_i}{pqr}, \end{equation}
where $A_i \in \A, B_i \in \B, C_i \in \C$ for every $i$ from $1$ to $n$ and, also, $p \in \F[a,b]$,
$q \in \F[a, c]$ and $r \in \F[b, c]$. Let us rewrite Equation~\eqref{f_definition} as $pqrf = 
\sum_{i=1}^n A_i B_i C_i$ with an additional condition that neither of $p$, $q$ and $r$
is zero: otherwise the denominator of the right-hand side of Equation~\eqref{f_definition} is zero.

For every formal power series of three variables $a$, $b$ and $c$ we can define its
\emph{trace}: such subset of $\N_0^3$, that a triple $(x, y, z)$ is in this subset if and only
if the coefficient of the series before $a^x b^y c^z$ is one. Traces of equal power series coincide.

How do the traces of left-hand and right-hand sides of equation $pqrf = \sum_{i=1}^n
A_i B_i C_i$ look like? Intuitively, the trace of the left-hand side should be near
the diagonal $x = y =z$ in its entirety, because $pqrf$ is a polynomial $pqr$, multiplied by
$f = \sum_{i=0}^{+\infty} a^i b^i c^i$. On the other hand, the trace of the
right-hand side has a ``block structure'': as we will establish later, it should be a finite
union of disjoint sets with type $X \times Y \times Z$. 

Our goal is to prove that such traces can coincide only if they are both finite.
This conclusion is quite natural: 
the trace of the left-hand side exhibits a ``high dependency'' between $x$, $y$ and $z$,
while the coordinates ``are almost independent'' in the trace of the right-hand side
(and they would be ``fully independent'' if there was only one set $X \times Y \times Z$
in the disjoint union). 

Let us proceed formally.

\begin{lm}\label{trace_blocky} The trace of the expression $\sum_{i=1}^n A_i B_i C_i$ is a finite disjoint
union of sets with type $X \times Y \times Z$.
\end{lm}
\begin{proof} For $x \in \N_0$, let us call the set of all such $i$ from $1$ to $n$, that
the coefficient of $A_i$ before $a^x$ is one, the \emph{a-type} of $x$. Similarly, 
define $b$-type and $c$-type.

Whether or not the triple $(x, y, z)$ is in the trace of $\sum\limits_{i=1}^n A_i B_i C_i$ 
depends only on the a-type of $x$, b-type of $y$ and c-type of $z$. Indeed, the coefficient
before $a^x b^y c^z$ is one in exactly such summands $A_i B_i C_i$, that the coefficient
of $A_i$ before $a^x$ is one, the coefficient of $B_i$ before $b^y$ is one and the coefficient
of $C_i$ before $c^z$ is one. Therefore the exact set of such summands depends only on
types of $x$, $y$ and $z$.

Consequently, the trace of $\sum_{i=1}^n A_i B_i C_i$ is a union of sets 
$X \times Y \times Z$, where $X$ is a set of numbers with some fixed $a$-type,
$Y$ is a set of numbers with some fixed $b$-type and $Z$ is a set of numbers with
some fixed $c$-type. There is only a finite number of such sets, because there is no more
than $2^n$ distinct a-types, no more than $2^n$ distinct b-types and no more than $2^n$ distinct $c$-types.
\end{proof} 

\begin{lm}\label{trace_diagonal} There exists a such constant $d$, that, for every triple $(x, y, z)$ from the trace of $pqrf$, the
conditions $|x - y| \leqslant d$, $|x - z| \leqslant d$ and $|y - z| \leqslant d$ hold.
\end{lm}
\begin{proof} Let $d$ be the degree of $pqr$ as of a polynomial of three variables. Because
$pqrf = pqr \cdot \sum_{i=0}^{+\infty} a^i b^i c^i$, the trace of $pqrf$ may only
contain triples $(\ell + i, m + i, k + i)$ for monomials $a^\ell b^m c^k$ from the polynomial $pqr$.
For such triples, $|x - y| 	= |\ell - m| \leqslant d$. Why? Because $d$ is the total degree of $pqr$ and, 
therefore, $0 \leqslant \ell \leqslant d$ and $0 \leqslant m \leqslant d$. Similarly,
$|x - z| \leqslant d$ and $|y - z| \leqslant d$.
\end{proof}

\begin{lm} If the traces of $\sum_{i=1}^n A_i B_i C_i$ and $pqrf$ coincide, then
they both are finite sets.
\end{lm}
\begin{proof} From Lemmata~\ref{trace_blocky} and~\ref{trace_diagonal}, a set that
is close to the diagonal coincides with a disjoint union of sets of a type $X \times Y \times Z$.
Then, each of the sets $X \times Y \times Z$ in the union is finite. Roughly speaking,
infinite sets of such type should contain elements that are 
arbitrarily far from the diagonal $x = y = z$.

Let us explain the previous paragraph more formally.  Indeed, suppose that one of the $X \times Y \times Z$ sets from the conclusion of Lemma~\ref{trace_blocky} is infinite. Then, at least one of the sets $X$, $Y$ and $Z$
is infinite. Without loss of generality, $X$ is infinite. Let $(x, y, z)$ be some element of 
$X \times Y \times Z$: it exists, because every infinite set contains at least one element.
Choose $x_{\text{new}}$ so $x_{\text{new}} > \max(y, z) + d$. Such $x_{\text{new}}$
exists, because $X$ is an infinite set of nonnegative integers. Then, $(x_{\text{new}}, y, z)
\in X \times Y \times Z$. Therefore, $(x_{\text{new}}, y, z)$ is in the trace of
$\sum_{i=1}^n A_i B_i C_i$. However, by Lemma~\ref{trace_diagonal}, 
$(x_{\text{new}}, y, z)$ cannot lie in the trace of $pqrf$, because $x_{\text{new}}$
differs from $y$ and $z$ too much.
\end{proof}
	
\begin{lm} The polynomial $1 + abc$ is irreducible as a polynomial over field $\F$.
\end{lm}
\begin{proof} 
Proof by contradiction: suppose that $1 + abc$ is reducible. Because its total degree is $3$,
it should split into a product of two polynomials with the total degrees $1$ and $2$ respectively.
In principle, enumerating all pairs of polynomials over $\F$ of total degree 
$1$ and $2$ on the computer does the job. For completeness, I will provide a proof
that does not use computer search.

Because the degree of $1 + abc$ with respect to each variable is $1$, each variable
occurs in exactly one of two factors --- if it occurs in both, the resulting degree is at least $2$,
if it occurs in neither,  the resulting degree is $0$. Hence, because the total degree of the second factor is $2$, but its degree in every variable is only $1$,
exactly $2$ variables occur in the second factor. 

Therefore, only one variable occurs in the 
first factor. Because the polynomial $1 + abc$ is symmetric with respect to the permutation
of variables, we may assume that the first factor depends only on $a$ and the second factor
depends only on $b$ and $c$.

The first factor is invertible, because the product is invertible. Previously, we have
shown that the first factor is a polynomial of $a$ of degree $1$, therefore the only
one possibility for the first factor remains: $1 + a$. The second factor is also 
invertible and is of degree $2$, therefore it is $1 + k_b b + k_c c + bc$ for some 
$k_b$ and $k_c$ from $\F$. Then, there is a summand $a \cdot 1 = a$ in their product,
which does not have anything to cancel up with. But their product is $1 + abc$, contradiction.
\end{proof}

Because $pqrf = \sum_{i=1}^n A_i B_i C_i$, the trace of $pqrf$ is finite. In other
words, $pqrf$ is a polynomial. Recall that $f = (1 + abc)^{-1}$, so  $\frac{pqr}{1 + abc}$
is a polynomial. Because the product of three polynomials $p$, $q$ and $r$ is divisible 
by an irreducible polynomial $1 + abc$, one of them is also divisible by $1 + abc$ (here we have used the fact that the ring $\F[a, b, c]$ of polynomials in three variables
is a unique factorization domain).
But this is impossible, because each of the polynomials $p$, $q$ and $r$ is non-zero 
(here we have finally used that condition from the statement of the lemma) and does not depend on one of the variables.

Finally, we have established the following theorem.

\begin{thm}\label{theorem_anbncn} The language $\set{a^n b^n c^n}{n \geqslant 0}$ is not described by 
a GF(2)-grammar.
\end{thm}

\begin{cor} The language $\set{a^n b^m c^\ell}{n = m \text{ or } m = \ell}$ is not
described by a GF(2)-grammar. 
\end{cor}
\begin{proof} Suppose that $\set{a^n b^m c^\ell}{n = m \text{ or } m = \ell}$ is 
described by a GF(2)-grammar. Then, $\set{a^n b^n c^n}{n \geqslant 0}$ also is, as the symmetric difference of  $\set{a^n b^m c^\ell}{n = m \text{ or } m = \ell}$ and
$\set{a^n b^m c^\ell}{n = m \text{ or } m = \ell, \text{ but not both}}$, where the latter
is described by a GF(2)-grammar~\cite[Example 2]{grammars-gf2}. Contradiction.
\end{proof}

\begin{cor} The language $\set{a^n b^m c^\ell}{n \neq m \text{ or } m \neq \ell}$
is not described by a GF(2)-grammar.
\end{cor}
\begin{proof} Otherwise, $\set{a^n b^n c^n}{n \geqslant 0} = (a^* b^* c^*) \triangle \set{a^n b^m c^\ell}{n \neq m \text{ or } m \neq \ell}$ 
would be described by a GF(2)-grammar as well.
\end{proof}

We have just proven that the language 
$\set{a^n b^m c^\ell}{n = m \text{ or } m = \ell}$ is not described by a GF(2)-grammar. 
Hence, it is inherently ambiguous. Previous proofs
of its inherent ambiguity were purely combinatorial, mainly based on Ogden's lemma, 
while our approach is mostly algebraic.

More importantly, we have proven that the language 
$\set{a^n b^m c^\ell}{n \neq m \text{ or } m \neq \ell}$ is not described by a GF(2)-grammar.
Therefore, it is inherently ambiguous. The inherent ambiguity of this language was a long-standing open question~\cite[p. 375]{autebert}.

\subsection{Other Applications}

There are many other applications to our techniques as well. For example, consider
the famous paper by Hibbard and Ullian about inherently ambiguous languages and 
languages that have a complement that can be recognized by an ordinary 
grammar~\cite{unamb-comp}.
Since the publication of that paper, a few improvements have been made.
For example, Maurer~\cite{Maurer} sharpened the statement of one of their theorems.
Also, recently
Martynova and Okhotin~\cite{unamb-linear-compl} found a language that is described by an unambiguous
\emph{linear} grammar,
but has a complement that cannot be described by any ordinary grammar at all.
I will suggest another potential improvement.

A significant part of Hibbard's and Ullian's paper is dedicated to proving that
the language $K \coloneqq \set{a^p b^q c^r d^s e^t}{\text{($p = q$ and $r = s$) or ($q = r$ and $s = t$)}}$
is inherently ambiguous.

With our techniques, the proof is simple. Indeed, $K$ is the symmetric difference of three languages
$\set{a^p b^q c^r d^s e^t}{\text{$p = q$ and $r = s$}} = \set{a^n b^n c^m d^m e^\ell}{n, m, \ell \geqslant 0}$, 
$\set{a^p b^q c^r d^s e^t}{\text{$q = r$ and $s = t$}} = \set{a^n b^m c^m d^\ell e^\ell}{n, m, \ell \geqslant 0}$ and
$\set{a^p b^q c^r d^s e^t}{\text{$p = q = r = s = t$}} = \set{a^n b^n c^n d^n e^n}{n \geqslant 0}$. The first two of these three languages are described by unambiguous grammars.
Therefore, if there exists an unambiguous grammar for $K$, then there exists a GF(2)-grammar for $\set{a^n b^n c^n d^n e^n}{n \geqslant 0}$. This is impossible by Theorem~\ref{theorem_anbncn}, because the language $\set{a^n b^n c^n d^n e^n}{n \geqslant 0}$ is ``even harder'' than $\set{a^n b^n c^n}{n \geqslant 0}$.

Intuitively, this is clear enough. Proving it formally with our current techniques can be a bit tricky. 
It is possible to prove this statement in an ``ad-hoc'' way, but let us use this opportunity
to prove some useful closure results (they are easy to prove, but have not been formally proved yet anywhere, because they take quite a lot of space to write down).
 
\begin{defn} For a language $K$ over an alphabet $\Sigma$ and a homomorphism $h \colon \Sigma \to \Omega^*$ we can define the \emph{GF(2)-homomorphic image} $h_{\oplus}(K)$ of $K$ under $h$ in the following way. A string $w \in \Sigma^*$ belongs to $h_{\oplus}(K)$ if and only if the set $h^{-1}(w) \cap K$ has an odd size. If this set is infinite for any $w \in \Sigma^*$, then $h_{\oplus}(K)$ is ill-defined.
\end{defn}
 
\begin{lm}\label{grammar-gf2-hom} Assume that a GF(2)-grammar $G$ describes a language $K$ over an alphabet $\Sigma$ and $h \colon \Sigma \to \Omega^*$ is a homomorphism. Then, there exists a GF(2)-grammar $G_h$ that describes the language $h_{\oplus}(K)$, as long as the following technical condition holds: for every  $w \in \Omega^*$, the total number of parse trees in $G$
for strings from $h^{-1}(w)$ is finite (this is a non-trivial condition if the set $h^{-1}(w)$ is infinite). In particular, $h_{\oplus}(K)$ is well-defined in that case.
\end{lm}
\begin{proof} Informally, $G_h$ is constructed by applying $h$ to every rule. Formally speaking, $G_h$ has the same set of nonterminals, the terminal alphabet $\Omega$ and,
for every rule $A \to w_0 \odot X_1 \odot w_1 \odot X_2 \odot \cdots \odot w_{\ell - 1} \odot X_\ell w_\ell$ in the original  GF(2)-grammar $G$, a corresponding rule $A \to h(w_0)  \odot X_1 \odot h(w_1) \odot X_2 \odot \cdots \odot h(w_{\ell - 1}) \odot X_\ell \odot w_\ell$ and no other rules (here, $w_i$ are strings over $\Sigma$ and $X_i$ are nonterminals of $G$).

Then, for every parse tree in $G$ for a string $w \in \Sigma^*$ we can construct a corresponding parse tree in $G_h$ for string $h(w)$ in the following way: when we apply some rule in the original parse tree, apply the corresponding rule in the corresponding parse tree. Moreover, this process is reversible.
Therefore, for a parse tree in $G_h$ for $w \in \Omega^*$ corresponds to a parse tree in $G$ for \emph{some} string from $h^{-1}(w)$. The exact string is uniquely determined by the sequence of used rules.

Hence, the number of parse trees in $G_h$ for a string $w \in \Omega^*$ is exactly
the total number of parse trees in $G$ for strings from $h^{-1}(w)$. In particular, $G_h$ is 
a valid GF(2)-grammar only if this number is always finite. 

If this number is finite, then there is only a
finite number of strings in $K \cap h^{-1}(w)$, because each one has at least one parse tree (by definition, strings in $K$ have ann odd number of parse trees in $G$, and every odd number is at least $1$). Hence, $h_{\oplus}(K)$ is well-defined. Finally, each string from $h^{-1}(w) \setminus K$ ``contributes'' an even number of parse trees in $G_h$ to $w$ and each string from $h^{-1}(w) \cap K$ ``contributes'' an odd number of parse trees in $G_h$ to $w$. Hence, $w$ is in $L(G_h)$ if and only if $h^{-1}(w) \cap K$ has odd size. Therefore, $L(G_h)$ is
exactly the language $h_{\oplus}(K)$.
\end{proof}

\begin{rem} Lemma~\ref{grammar-gf2-hom} is somewhat unconventional (and can be tricky to use) in the following way: usually we prove closure results for the languages themselves without any assumptions about the underlying grammar. In Lemma~\ref{grammar-gf2-hom}, on the other hand, we need to know some properties of the original GF(2)-grammar. Let us prove a weaker, but easier to use corollary that concerns the languages themselves.
\end{rem}

\begin{lm}\label{lang-gf2-hom} If a language $K$ over an alphabet $\Sigma$ is described by a GF(2)-grammar and $h \colon \Sigma \to \Omega^*$ is a non-erasing homomorphism, then $h_{\oplus}(K)$ is described by a GF(2)-grammar as well.
\end{lm}
\begin{proof} Because $h$ is non-erasing, each string $w \in \Omega^*$ has only a finite number of preimages under $h$. Each of those preimages has a finite number of parse trees by the definition 
of a GF(2)-grammar. Hence, the technical condition from Lemma~\ref{grammar-gf2-hom} is
satisfied. 
\end{proof}
 
 \begin{cor} The set $M := \set{a^n b^n c^n d^n e^n}{n \geqslant 0}$ is not described by a GF(2)-grammar.
 \end{cor}
 \begin{proof} Consider a homomorphism $h$, defined by $h(a) = aa, h(b) = b, h(c) = b, h(d) = c, h(e) = c$. It is a non-erasing homomorpism. Therefore, we can apply Lemma~\ref{lang-gf2-hom} and deduce that $h_{\oplus}(M) = \set{a^{2n} b^{2n} c^{2n}}{n \geqslant 0}$ is described by a GF(2)-grammar ($h_{\oplus}(M)$ looks like this, because $h$ is injective on $M$). Now, let us apply an deterministic finite transducer
that halves the lengths of all blocks of consequent equal letters and does not return anything at all if there are any such blocks of odd length. Any DFT is an injective NFT. Therefore, by
closure under injective NFTs~\cite[Theorem 17]{grammars-gf2}, there is a GF(2)-grammar for $\set{a^n b^n c^n}{n \geqslant 0}$. Contradiction.
\end{proof}

\begin{rem} Lemma~\ref{lang-gf2-hom} generalizes the ``method of unary image'' from the original paper about expressive powers of GF(2)-grammars~\cite[Theorem 8]{grammars-gf2}. It is possible to
generalize even further, by defining GF(2)-transducers and proving the equivalents of Lemmata~\ref{grammar-gf2-hom} and~\ref{lang-gf2-hom} for them, but we will avoid that for now, because both the statement and the proof are quite technical.
\end{rem}

\section{General Bounded Languages}

Until now, I have only considered the case of \emph{letter-bounded} languages. Recall that the usual definition of a bounded language
is more general: a language is bounded if it is a subset of $w_1^* w_2^* \ldots w_k^*$ for some strings $w_1$, $w_2$, $\ldots$, $w_k$ (clearly, we may assume that all $w_i$ are non-empty).
Notice that there is \emph{no} requirement that each string from the language has an exactly one representation of the form $w_1^{\ell_1} w_2^{\ell_2} \ldots w_k^{\ell_k}$. 

Suppose that we have already fixed both the integer $k$ and the sequence $w_1$, $w_2$, $\ldots$, $w_k$ of the strings over $\Sigma$ 
and are now considering a language $L \subseteq w_1^* w_2^* \ldots w_k^*$. Then, for a string $s$ from $w_1^* w_2^* \ldots w_k^*$ let us say that the tuple of nonnegative integers $(\ell_1, \ell_2, \ldots, \ell_k)$ is a \emph{representation} of $s$,
if $s = w_1^{\ell_1} w_2^{\ell_2} \ldots w_k^{\ell_k}$. Clearly, some strings have many representations. For example, if $k = 3$, $w_1 = ab$, $w_2 = a$ and $w_3 = ba$,
then the string $(ab)^n a$ has exactly $n+1$ representations: $(0, 1, n)$, $(1, 1, n-1)$, $\ldots$, $(n, 1, 0)$. From here on, I will sometimes
abuse the notation a bit and write $w_1^{\ell_1} w_2^{\ell_2} \ldots w_k^{\ell_k}$ and mean the representation $(\ell_1, \ell_2, \ldots, \ell_k)$ of the string  $w_1^{\ell_1} w_2^{\ell_2} \ldots w_k^{\ell_k}$ and not the string  $w_1^{\ell_1} w_2^{\ell_2} \ldots w_k^{\ell_k}$ itself.

\begin{thm}\label{general-bounded} Let $k$ be a positive integer and $w_1$, $w_2$, $\ldots$, $w_k$ be some strings over an alphabet $\Sigma$. Then, if a language $L \subseteq w_1^* w_2^* \ldots w_k^*$ can be described by a GF(2)-grammar, then the language $L_a \coloneqq \set{a_1^{\ell_1} a_2^{\ell_2} \ldots a_k^{\ell_k}}{w_1^{\ell_1} \ldots w_k^{\ell_k} \in L}$ over the alphabet $\Omega \coloneqq \{a_1, \ldots, a_k\}$ can also be described by a GF(2)-grammar. In particular, the formal power seriers $\sum\limits_{w_1^{\ell_1} \ldots w_k^{\ell_k} \in L} a_1^{\ell_1} \ldots a_k^{\ell_k}$ is in $R_{a_1, a_2, \ldots, a_k}$. Moreover, the construction is effective: there exists an algorithm that
constructs a GF(2)-grammar for $L_a$ when the number $k$, the strings $w_1$, $w_2$, $\ldots$, $w_k$ and a GF(2)-grammar for the language $L$ are given in the input.
\end{thm} 
\begin{proof} Consider a nondeterministic finite tranducer that does the following. It has $k$ ``important'' states $q_1$, $q_2$, $\ldots$, $q_k$ (depending on the exact definition of a nondeterministic transducer, it may need more states for ``bookkeeping''), with $q_1$ being the sole starting state. When in the state $q_i$, it nondeterministically chooses to either
\begin{enumerate}
\item move to the state $q_{i+1}$ without reading or writing anything, 
\item or stay in the state $q_i$, read the string $w_i$ and write the single character $a_i$ to the output. 
\end{enumerate} 
In other words, on a string $s \in w_1^* w_2^* \ldots w_k^*$, it ``guesses'' a representation $s = w_1^{\ell_1} w_2^{\ell_2} \ldots w_k^{\ell_k}$ and ``translates'' the representation to $a_1^{\ell_1} \ldots a_k^{\ell_k}$ by replacing each $w_i$ with the corresponding $a_i$.

Now, each string from $\Sigma^*$ can have multiple images from $\Omega^*$: one for each representation of the string. However, the transduction is still injective: while each string $s \in \Sigma^*$ can have multiple images in $\Omega^*$, with an image $a_1^{\ell_1} a_2^{\ell_2} \ldots a_k^{\ell_k}$ corresponding to the representation $s = w_1^{\ell_1} \ldots w_k^{\ell_k}$, each string from $\Omega^*$ has an exactly one preimage from $\Sigma^*$. As shown previously by Makarov and Okhotin~\cite{grammars-gf2}, the class of languages described by GF(2)-grammars is closed under injective nondeterministic finite transductions. Hence, $L_a$ is described by some GF(2)-grammar. 

Finally, the construction is effective because the result of applying an injective NFT to a given GF(2)-grammar can be obtained by an explicit construction by Makarov and Okhotin~\cite{grammars-gf2}, similarly to the simpler case of intersecting a GF(2)-grammar with a regular language. 
Here, it is important that Eilenberg's construction~\cite{eilenberg} (used as one of the steps of the aforementioned argument by Makarov and Okhotin) that transforms a single-valued NFT into an unambiguous NFT is also known to be effective.
\end{proof}

While the statement of Theorem~\ref{general-bounded} is true, it is inherently
one-sided: the existence of a GF(2)-grammar for $L \subseteq w_1^* w_2^* \ldots w_k^*$ implies the existence of a GF(2)-grammar for its ``image'' $L_a \coloneqq \set{a_1^{\ell_1} a_2^{\ell_2} \ldots a_k^{\ell_k}}{w_1^{\ell_1} w_2^{\ell_2} \ldots w_k^{\ell_k} \in L}$, but the existence of a GF(2)-grammar for $L_a$ \emph{does not} necessarily imply the existence of a GF(2)-grammar for $L$. 

The main issue with Theorem~\ref{general-bounded} is that it considers all
representations of a string at the same time, regardless of how many of them exist. Indeed, if $s$ is in $L$, then each representation $s = w_1^{\ell_1} \ldots w_k^{\ell_k}$ has a corresponding string $a_1^{\ell_1} \ldots a_k^{\ell_k}$ in the language $L_a$; if $s$ is not in $L$, then there are no strings in the language $L_a$ that correspond to the representations of $s$. Arguably, this choice of representations does not fit the spirit of GF(2)-grammars
and GF(2)-operations. Indeed, if we tried to apply Lemma~\ref{lang-gf2-hom} to the language $L$ with the natural homomorphism
$h$ defined by $h(a_i) = w_i$, we would run into the issue that $h_{\oplus}(L_a)$ does not necessarily coincide with $L$: the strings from $L$ with an odd number of representations are in the $h_{\oplus}(L_a)$, but the strings with an even number of representations are not.

However, there is a well-known way to deal with this issue, suggested by Ginsburg and Ullian~\cite{unamb-class}. 
\begin{defn}
The \emph{leftmost representation} (Ginsburg and Ullian use the term ``$w$-maximal'' instead) of the string $s$ is the representation $s = w_1^{\ell_1} \ldots w_k^{\ell_k}$ that lexicographically maximizes
the tuple $(\ell_1, \ell_2, \ldots, \ell_k)$. In other words, it chooses the representation with the largest $\ell_1$. If there
are ties, they are broken by $\ell_2$: the representation with the largest $\ell_2$ is chosen. If there are still some ties, they
are broken by $\ell_3$, and so on. 
\end{defn}
\begin{oldtheorem}[{\cite[Lemma 5.2]{unamb-class}}]\label{greedy} For an integer $k$ and non-empty strings $w_1$, $w_2$, $\ldots$, $w_k$ consider the \emph{language of the leftmost representations of the strings from the (regular) language $w_1^* w_2^* \ldots w_k^*$}. Formally speaking, consider the language
$R \coloneqq \set{a^{\ell_1} a^{\ell_2} \ldots a^{\ell_k}}{\text{$(\ell_1, \ell_2, \ldots, \ell_k)$ is the leftmost representation of the string $w_1^{\ell_1} w_2^{\ell_2} \ldots w_k^{\ell_k} \in w_1^* w_2^* \ldots w_k^*$}}$ over the alphabet $\{a_1, a_2, \ldots, a_k\}$. Then, the language $R$ is regular. Moreover, the construction is effective: there exists an algorithm that constructs some DFA for the language $R$ if the number $k$ and the strings $w_1$, $w_2$, $\ldots$, $w_k$ are given in the input.
\end{oldtheorem}

By combining Theorem~\ref{general-bounded} with Theorem~\ref{greedy} we get the following result.
\begin{thm}\label{full-bounded}
A language $L \subseteq w_1^* w_2^* \ldots w_k^*$ is described by a GF(2)-grammar if and only if the language $L_{\ell} \coloneqq \set{a_1^{\ell_1} \ldots a_k^{\ell_k}}{\text{there exists a string $s \in L$ such that $(\ell_1, \ell_2, \ldots, \ell_k)$ is the leftmost representation of $s$}}$ is described by a GF(2)-grammar. Moreover, the construction is effective in both directions in the following sense:
\begin{enumerate}
\item there exists an algorithm that constructs some GF(2)-grammar for the (letter-bounded) language $L_{\ell}$ if the number $k$, the strings $w_1$, $w_2$, $\ldots$, $w_k$ and a GF(2)-grammar for the (bounded) language $L$ are given in the input;
\item there exists an algorithm that constructs some GF(2)-grammar for the (bounded) language $L$ if the number $k$, the strings $w_1$, $w_2$, $\ldots$, $w_k$ and a GF(2)-grammar for the (letter-bounded) language $L_{\ell}$ are given in the input.
\end{enumerate}
\end{thm}
\begin{proof} Essentially, the language $L_{\ell}$ is the language of the \emph{leftmost} representations of the strings from $L$. Consider the language $L_a \coloneqq \set{a_1^{\ell_1} a_2^{\ell_2} \ldots a_k^{\ell_k}}{w_1^{\ell_1} \ldots w_k^{\ell_k} \in L}$ of \emph{all} representations of the strings from the language $L$. Then, $L_{\ell} = L_a \cap R$, where $R \coloneqq \set{a^{\ell_1} a^{\ell_2} \ldots a^{\ell_k}}{\text{$(\ell_1, \ell_2, \ldots, \ell_k)$ is the leftmost representation of the string $w_1^{\ell_1} w_2^{\ell_2} \ldots w_k^{\ell_k} \in w_1^* w_2^* \ldots w_k^*$}}$ over the alphabet $\{a_1, a_2, \ldots, a_k\}$ is the language of the leftmost representations of all strings from \emph{the whole} $w_1^* w_2^* \ldots w_k^*$. 

Suppose that $L$ is described by a GF(2)-grammar. Then, by Theorem~\ref{general-bounded}, the language $L_a$ is described by a GF(2)-grammar as well. Moreover, by Theorem~\ref{greedy}, the language $R$ is regular. Hence, the language $L_a \cap R = L_{\ell}$ is also described by a GF(2)-grammar.

Conversely, suppose that $L_{\ell}$ is described by a GF(2)-grammar. For each $s \in L$, there exists a unique representation of $s$ in $L_{\ell}$ (exactly the leftmost representation). Consider the homomorphism $h$ that maps each character $a_i$ to the corresponding string $w_i$. Because $h$ is non-erasing, Lemma~\ref{lang-gf2-hom} implies that the language $h_{\oplus}(L_{\ell})$ is described by a GF(2)-grammar. However,
$h_{\oplus}(L_{\ell})$ is exactly the language $L$, because each string from $L$ has an exactly one preimage from $L_{\ell}$ with respect to $h$ (said preimage corresponds to the leftmost representation of $s$). Hence, $L$ is described by a GF(2)-grammar.

The constructions are effective, because every single step of the process is effective. Indeed, Theorems~\ref{general-bounded} and~\ref{greedy} are
effective because that is a part of their statement. Intersecting a GF(2)-grammar with a regular language given by an explicit DFA is effective.
Finally, going from $L_{\ell}$ to $h_{\oplus}(L_{\ell})$ is also effective, because the proof of Lemma~\ref{grammar-gf2-hom} is an
explicit construction.
\end{proof}

Essentially, Theorem~\ref{full-bounded} allows us to reduce the case of arbitrary bounded languages to the case of letter-bounded languages. The exact nature of
the leftmost representations is not too important here; the only important things here are that we have chosen a single ``canonical'' representation for each string from $w_1^* \ldots w_k^*$ and that the ``language of all canonical representations'' is regular. 

One interesting implication of Theorem~\ref{full-bounded} is that, for $k = 2$, all
subsets of $w_1^* w_2^*$ that can be described by a (possibly ambiguous) ordinary grammar can be described by a 
GF(2)-grammar as well. This follows from the result of Ginsburg and Ullian, that implies that no inherently ambiguous
subsets of $w_1^* w_2^*$~\cite{unamb-class} (that is, every subset of $w_1^* w_2^*$ that can be described by an ordinary grammar 
can be described by an unambiguous grammar as well). 

\begin{rem}
It is important to note that all results in this section explicitly use the fact that a bounded language $L$ is specified by a GF(2)-grammar and
the sequence $w_1$, $w_2$, $\ldots$, $w_k$ of the strings, such that $L$ is a subset of $w_1^* w_2^* \ldots w_k^*$. This is important:
all the results above are only effective when the strings $w_1$, $w_2$, $\ldots$, $w_k$ are explicitly given. Indeed, almost all steps in the proofs require the knowledge of the strings $w_1$, $w_2$, $\ldots$, $w_k$ to be effectively performed. 

However, when talking about the family of bounded languages described by ordinary grammars, it makes sense to just say ``a bounded language'' instead of exactly specifying $w_1$, $w_2$, $\ldots$, $w_k$. Indeed, as shown by Ginsburg and Spanier~\cite[Theorem 5.2]{bounded-original},
for any \emph{ordinary} grammar $G$, there is a procedure that either (correctly) says that $L(G)$ is not bounded, or returns some strings $w_1$, $w_2$, $\ldots$, $w_k$ such that $L(G) \subseteq w_1^* \ldots w_k^*$. 

The same cannot be said about GF(2)-grammars. Indeed,
it is unknown if the emptiness testing for a GF(2)-grammar is decidable. Consequently, if $G$ is a GF(2)-grammar over an alphabet $\Sigma$
and $a$ and $b$ are two new characters that are not in $\Sigma$, then the natural GF(2)-grammar for the language $L(G) \{a, b\}^*$ either describes the empty language (if $L(G)$ is empty) or a not bounded one (if $L(G)$ is not empty). Hence, there is no currently known procedure to even decide whether or not a given GF(2)-grammar describes a bounded language, meaning that knowing $w_1$, $w_2$, $\ldots$, $w_k$ beforehand is important.
\end{rem}

\section{Are the Converse Statements True?}\label{converse}

It would be interesting to know whether the converse statements to 
Theorems~\ref{Main_ab},~\ref{Main_abc} and~\ref{Main_many} are true.
Of course, the exact converse statement to Theorem~\ref{Main_ab} is wrong on the 
technicality that not all elements of $R_{a, b}$ are power series. However, if we restrict
ourselves to $R_{a, b} \cap \F[[a, b]]$, we get the following conjecture, which I believe in:
\begin{conjecture}\label{conj_ab} If $f \in R_{a, b} \cap \F[[a, b]]$, then the language 
$\Dual^{-1}(f) \subseteq a^* b^*$ can be described by a GF(2)-grammar.
\end{conjecture}
The following Theorem~\ref{ab_part_conv} is an evidence in favor of Conjecture~\ref{conj_ab}.
We will need the following definition to state Theorem~\ref{ab_part_conv}.
\begin{defn}\label{ab_int_def} Denote by $R_{a, b}^{int}$ the set of Laurent series 
from $\F((a, b))$ that can be represented as $\frac{\sum_{i=1}^n A_i B_i}{p}$, where $n \geqslant 0$, $A_i \in \A$ and $B_i \in \B$ for $i$ in range from $1$ to $n$, and
$p \in \poly(a, b)$ is a polynomial with constant term equal to $1$.
\end{defn}
\begin{rem} The only difference between the definitions of $R_{a,b}$ and $R_{a,b}^{int}$
is that the denominator $p$ of the fraction is required to be invertible as an element of $\F[[a, b]]
\supset \poly(a, b)$. In particular, $R_{a, b}^{int}$ is a subset of $\F[[a, b]]$. Also, by 
definition, $R_{a, b}^{int} \subseteq R_{a, b}$.
\end{rem}
\begin{thm}\label{ab_part_conv} If $f \in R_{a,b}^{int}$, then
$\Dual^{-1}(f)$ can be described by a GF(2)-grammar.
\end{thm}
\begin{proof}[Sketch of the proof] Suppose that $f = \left(\sum_{i=1}^n A_i B_i \right) / p$, as
in Definition~\ref{ab_int_def}, and $p = 1 + \sum_{j=1}^d a^{k_j} b^{\ell_j}$, where $k_j + \ell_j > 0$. Let $S$ be the starting symbol of some GF(2)-grammar
that describes $\Dual^{-1}(\sum_{i=1}^n A_i B_i)$ (such GF(2)-grammar exists by closure properties). Add a new starting symbol $S_{\mathrm{new}}$ and a new rule $S_{\mathrm{new}} \to (\oplus_{j=1}^d a^{k_j} S_{\mathrm{new}} b^{\ell_j}) \oplus S$ to the GF(2)-grammar.
The new GF(2)-grammar describes $\Dual^{-1}(f)$. The condition $k_j + \ell_j > 0$ is important, because it assures
that all strings still have only a finite number of parse trees.
\end{proof}

It is tempting to say that the converse to Theorem~\ref{ab_part_conv} is true as well. In particular, it may be tempting to 
conjecture that $R_{a,b}^{int} = R_{a,b} \cap \F[[a,b]]$. If that were true, the converse to Theorem~\ref{ab_part_conv} would 
immediately follow. Unfortunately, neither is true.

Let us disprove that $R_{a, b}^{int} = R_{a, b} \cap \F[[a, b]]$ by showing an element of $R_{a, b} \cap \F[[a, b]]$ that is not 
an element of $R_{a, b}^{int}$. Define $f = \sum\limits_{n=1}^{+\infty} a^{2^n}$ and $g = \sum\limits_{n=1}^{+\infty} b^{2^n}$.
Then, $\dfrac{f + g}{a + b} = \sum\limits_{n = 1}^{+\infty} \left( \sum\limits_{m=0}^{2^n - 1} a^m b^{2^n - 1 - m} \right)$, because,
in characteristic $2$, it is true that $\dfrac{a^k + b^k}{a + b} = a^{k-1} + a^{k-2} b + \ldots + b^{k - 2} a + b^{k - 1}$ for any $k \geqslant 1$.
Hence, $h \coloneqq \dfrac{f + g}{a + b}$ is an element of $R_{a, b} \cap \F[[a, b]]$.

\begin{lm} The formal power series $h$ is not an element of $R_{a, b}^{int}$.
\end{lm}
\begin{proof} Proof by contradiction. Assume that $h = \dfrac{\sum_{i=1}^d A_i B_i}{1 + ap + bq}$ for some integer $d$, $A_i \in \A$,
$B_i \in \B$ and $p, q \in \F[a, b]$. Then, $\sum_{i=1}^d A_i B_i = (1 + ap + bq) \cdot \sum\limits_{n = 1}^{+\infty} \left( \sum\limits_{m=0}^{2^n - 1} a^m b^{2^n - 1 - m} \right)$. Consider a large enough $n$ and the coefficients of the right-hand side on the ``diagonal''
$a^x b^y$ with $x + y = 2^n - 1$. They are all $1$, while the coefficients on the neighbouring ``diagonal'' $a^x b^y$ with $x + y = 2^n - 2$ are all $0$. Informally speaking, similarly to Theorem~\ref{theorem_anbncn}, this is does not work well with the ``blocky'' structure of $\sum_{i=1}^d A_i B_i$. 

Formally, let us pick a number $r$ such that the degree of $1 + ap + bq$ does not exceed $2^{r-1}$. Then, consider the power
series on the left-hand side, but zero out all the coefficients before $a^x b^y$, such that $x + y$ does not have remainder $2^r - 1$ modulo $2^r$. Let us denote such series by $h'$.
In the grammar terms, it is equivalent to intersecting the the language with the regular language of strings with length that is 
$2^r - 1$ modulo $2^r$. Then, the right-hand side ``blows up'' in size, but retains its form. Indeed, for each $i$ from $1$ to $d$ and each $s$ from $0$ to $2^r - 1$ let us define $A'_{i, s}$ which is exactly the sum of all such $a^j$, that the coefficient before $a^j$ in $A_i$ is $1$
and the remainder of $j$ modulo $2^r$ is $s$. By Christol's theorem, all $A'_{i, s}$ are elements of $\A$. Define $B'_{i, s} \in \B$ similarly. 
Then, $\sum\limits_{i=1}^d A_i B_i = \sum\limits_{i=1}^d \sum\limits_{s = 0}^{2^r - 1} \sum\limits_{t = 0}^{2^r - 1} A'_{i, s} B'_{i, t}$
and the result of only considering the ``correct'' coefficients $a^x b^y$ is the sum $\sum\limits_{i=1}^d \sum\limits_{s=0}^{2^r - 1} A'_{i, s} B'_{i, 2^r - 1 - s}$. 

Now, let us draw a $2^n \times 2^n$ matrix $C$, with its its rows being numbered from $0$ to $2^n - 1$ and its columns also
being numbered from $0$ to $2^n - 1$. Let the entry $C_{x, y}$ be the coefficient of $h'$ in $a^x b^y$. Then, $C_{x, y} = 1$
for $x + y = 2^n - 1$. But, even more importantly, $C_{x, y} = 0$ for $x + y > 2^n - 1$. Indeed, the ``diagonals'' with $x + y$ being at most $2^{n-1} - 1$ do not affect the entries in the lower-right triangle due to the high distance, the ``diagonal'' $x + y = 2^n - 1$ does not affect the triangle because the of the requirement on remainders (as long as $n \geqslant r$) and the ``diagonals'' with $x + y \geqslant 2^{n + 1} - 1$ do not affect the triangle, because, even for the lower-right corner of the matrix, the sum is $2 \cdot (2^{n} - 1) < 2^{n+1} - 1$. Hence, the matrix $C$ is of full rank $2^n$. 
On the other hand, the matrix $C$ also represents the power series $\sum\limits_{i=1}^d \sum\limits_{s=0}^{2^r - 1} A'_{i, s} B'_{i, 2^r - 1 - s}$. Due to its ``blocky structure'', each summand in this sum is represented by a matrix of rank at most $1$. Hence, the matrix that corresponds to the whole sum has the rank at most $d \cdot 2^r$. By picking a large enough $n$, we arrive to a contradiction.
\end{proof}

Now, let us show that the converse of Theorem~\ref{ab_part_conv} is also not true.
\begin{lm} The language $\Dual^{-1}(h)$ can be described by a GF(2)-grammar.
\end{lm}
\begin{proof} It is easy to see that $f^2 + f + a = 0$ and $g^2 + g + b = 0$. Hence, $(f + g)^2 = f^2 + g^2 = (f + a) + (g + b) = (f + g) + (a + b)$. Therefore, $\dfrac{f + g}{a + b} = \dfrac{1}{f + g + 1}$. Indeed, $(f + g)(f + g + 1) = (f + g)^2 + (f + g) = ((f + g) + (a + b)) + (f + g) = a + b$. Hence, the grammar $S \to AS \oplus SB$ describes the language $\Dual^{-1}(h)$, where $A$ is a nonterminal that describes the language $\set{a^{2^n}}{n \geqslant 1}$ and $B$ is a nonterminal that describes the language $\set{b^{2^n}}{n \geqslant 1}$. Indeed, the grammar corresponds to the equation $\Dual(L(S)) = f \cdot \Dual(L(S)) + \Dual(L(S)) \cdot g$, or $\Dual (L(S)) = \dfrac{1}{f + g + 1} = h$.

If the algebraic tricks do not convince you, it is possible to verify the GF(2)-grammar by directly computing the number of parse trees for each $a^x b^y$ as well. However, it is much more difficult.
\end{proof}

In some sense, Theorem~\ref{ab_part_conv} shows that the converse to Theorem~\ref{Main_ab} is ``almost correct''.
On the other hand, the converse to Theorem~\ref{Main_many} seems to be very from being true, because it fails to take ``overlapping requirements'' into account. Indeed, the language $\set{a^n b^m c^n d^m}{n, m \geqslant 0}$
is conjectured to not be described by any GF(2)-grammar. However,
said language does not violate the conclusion of Theorem~\ref{Main_many},
because $\sum\limits_{n=0}^{+\infty} \sum\limits_{m=0}^{+\infty} a^n b^m c^n d^m = \dfrac{1}{(1+ab)(1+cd)}$.

\section{Conclusion}\label{section_conclusion}

Let us make some concluding remarks and discuss some possible future developments.

Firstly, note that it took us roughly the same effort to prove the inherent ambiguity of $\set{a^n b^m c^{\ell}}{n = m \text{ or } m = \ell}$
and $\set{a^n b^m c^{\ell}}{n \neq m \text{ or } m \neq \ell}$, despite the former being a textbook example of inherently ambiguous language
and the latter not being known to be inherently ambiguous before.
Intuitively, it is very difficult to capture weak conditions like inequality using Ogden's lemma, while our approach can replace inequality with
a strong condition (equality) by taking the complement. 

Secondly, the proofs of Theorems~\ref{Main_ab} and~\ref{Main_abc} 
start similarly to the reasoning Ginsburg and Spanier used to characterize
bounded languages described by ordinary grammars~\cite{bounded-original, bounded-semi-linear}, 
but diverge after taking some steps. This is not surprising; ordinary grammars
have good monotonicity properties (a string needs only one parse tree to be in the language), 
but bad algebraic properties (solving systems of language equations is much harder than solving
systems of linear equations). In GF(2)-grammars, it is the other way around: there are no
good monotonicity properties, but algebraic properties are quite remarkable. 

\paragraph{Most recent related work}
Since the publication of the conference version of this paper~\cite{Makarov_DLT},
Koechlin~\cite{koechlin} presented a different proof of the inherent ambiguity
of the language $L_2 = \set{a^n b^m c^{\ell}}{n \neq m \text{ or } m \neq \ell}$
by giving an algebraic characterization of bounded
languages described by unambiguous grammars directly, without going through GF(2)-grammars first.
Essentially, he directly transforms Ginsburg's and Ullian's result~\cite[Theorems 5.1 and 6.1]{unamb-class} about the structure of bounded
languages described by an unambiguous grammar into a very simple algebraic description of the underlying semlinear set. Therefore, his method avoids all the issues mentioned in the introduction (the ``density'' of $L_2$ and the fact that $L_2$ has an algebraic generating function). 
I recommend giving it a read, if you liked my paper. I think that it is very insightful and 
enlightening.

\paragraph{Future research}
Perhaps, our methods could be used to make some progress on the equivalence problem for unambiguous grammars.
Indeed, the equivalence problem for unambiguous grammars is closely related to the emptiness problem for GF(2)-grammars.
If it is decidable, whether GF(2)-grammar describes an empty language or not, then the equivalence of unambiguous grammars is decidable as well.
If it is not, the proof will most probably shed some light on the case of unambiguous grammars anyway.
However, resolving the emptiness problem for GF(2)-grammars in one way or another still seems to be out of reach.

Understanding how our methods relate to the analytic methods of Flajolet~\cite{flajolet}, is another interesting question.
One can see Theorem~\ref{unary-alphabet} as 
an alternative formulation of Christol's theorem~\cite{christol} for $\F$ specifically,
that involves GF(2)-grammars instead of 2-automatic sequences on the ``combinatorial side''.
Then, Christol's theorem can be seen as a finite field analogue of 
Chomsky-Schutzenberger enumeration theorem,
because both relate counting properties of different grammar families to algebraic power series
over fields by ``remembering'' only the length of the string, but nothing else:
\begin{oldtheorem} [Chomsky-Schutzenberger enumeration theorem~\cite{unambiguous-algebraic}] 
\label{unamb-gen-func}
If $L$ is a language
described by an unambiguous grammar, and $a_k$ is the number of strings of length
$k$ in $L$, then the power series $\sum\limits_{k=0}^{+\infty} a_k x^k$ is algebraic
over $\Q[x]$.
\end{oldtheorem}
\begin{oldtheorem}[Christol's theorem for $\F$~\cite{christol, grammars-gf2}; usually not stated this way]
\label{christol-as-grammars}
If $L$ is a language described by a GF(2)-grammar,
and $a_k$ is the number of strings of length $k$ in $L$,
then the power series $\sum\limits_{k=0}^{+\infty} (a_k \bmod 2) \cdot x^k$ is algebraic 
over $\F[x]$. 
\end{oldtheorem}
\begin{rem}
Technically speaking, Theorem~\ref{christol-as-grammars} is much weaker than 
Christol's theorem, because it is stated in only one direction (and the easier one to boot). I wanted to highlight the similarity between Theorems~\ref{unamb-gen-func} and~\ref{christol-as-grammars}, so I intentionally avoided stating the converse implication in Theorem~\ref{christol-as-grammars}.
\end{rem}

This similarity gives us some hope that our methods can be at least partially transferred to the analytic setting. Moreover, a lot (though not all) of the arguments used in our work can be modified to work over an arbitrary field.

\section*{Acknowledgements} This work was performed at the Saint Petersburg Leonhard Euler International Mathematical Institute and supported by the Ministry of Science and Higher Education of the Russian Federation (agreement no. 075-15-2022-287)

\newpage
\begin{appendices}
\section{What is This Part of the Paper about?}\label{alg_intro}
In Sections~\ref{alg_intro}--\ref{alg_expr_abc} I present a longer, but a more elementary way to prove Theorems~\ref{Main_ab_old} and~\ref{Main_ab_oldc},
through the use of \emph{algebraic expressions}. In the end, it is no surprise that we will end up essentially reproving many of the algebraic statements I use in the 
main body of the paper, but in our particular special case. 

So, was it all in vain? I believe that the answer is ``No''. In a sense, this part of the paper ``demystifies'' its main body, because it more closely follows my
original pattern of thinking. If you think that the new argument is too ``magical'', then, probably, this part of the paper is for you. What follows is, at least
in my opinion, a very intuitive line of thinking. And in the end, we will be just one simple, but brilliant observation away from the much simpler final version
of the argument.

The credit for this observation goes to an anonymous reviewer from MFCS 2020 conference. In short, and very paraphrased, 
they said ``You are essentially proving that
$\dfrac{\sum \A\B}{\poly(a, b)}$ describes a field, aren't you? You can prove this in two sentences by using standard results about field extensions''.
When pointed out, it sounds very simple, but, funnily enough, I have never considered thinking about the equivalence between algebraic expressions 
$\dfrac{\sum \A\B}{\sum \A\B}$ and $\dfrac{\sum \A\B}{\poly(a, b)}$ in such a way before!

So, in the end, you can see the following part of the paper as both an explanation for more ``magical'' part of the final argument, and as a story
with a ``plot twist'' that almost everything I did was unnecessary. Personally, I think that knowing the twist only makes it more interesting to observe.

On the other hand, this part of the paper has absolutely no new results that are not proved in
the main body of the text. Hence, if you neither find abstract algebraic arguments too ``magical'' nor are interested in the ``history'' of the main results, you may safely skip it.

\section{Subsets of \texorpdfstring{$a^*b^*$}{a*b*} through Algebraic Expressions}\label{alg_expr_ab}
\subsection{Algebraic Expressions}

Let us define the meaning of words ``Laurent series $f$ match
algebraic expression $F$''. 

Informally, algebraic expressions are some formulas of symbols $\A$, $\B$, $\poly(a,b)$
and $\rat(a,b)$ that use additions, multiplications, divisions and ``finite summation'' operator,
denoted by $\sum$. Here, $\poly(a, b)$ denotes the set $\F[a, b]$ of polynomials in variables
$a$ and $b$. Similarly, $\rat(a, b)$ denotes the set $\F(a, b)$ of rational functions in variables
$a$ and $b$.

Several examples of \emph{algebraic expressions}: 
$\rat(a, b)$, $\B$, $\sum \A$, $\dfrac{\sum \A\B}{\sum \A\B}$, $\sum \dfrac{\sum \A\B}{\rat(a,b)}$, 
$\sum \A \B \rat(a, b)$.

\begin{defn}
Laurent series $f$ \emph{match} algebraic expression $F$ 
if and only if $f$ can be obtained from 
$F$ by substituting elements of $\poly(a, b)$, $\rat(a, b)$, $\A$ and $\B$ 
for the corresponding symbols 
(not necessarily the same elements for the same symbols).
The construct $\sum G$ 
corresponds to a finite, possibly empty, sum of Laurent series, with every summand matching $G$.
\end{defn}

\begin{ex} The set of Laurent series matching $\sum \A\B$ 
is exactly the set of all power series 
representable as $A_1 B_1 + \ldots + A_n B_n$, 
where $n$ is any nonnegative integer, 
and $A_i \in \A$, $B_i \in \B$ for every $i$ from $1$ to $n$ inclusive.
\end{ex}

\begin{ex} All rational functions in variables $a$ and $b$, and only them, match
algebraic expressions $\rat(a, b)$ and $\dfrac{\poly(a, b)}{\poly(a, b)}$. 
\end{ex}

\begin{ex} Laurent series match $\sum \dfrac{\sum \A\B}{\rat(a, b)}$ if and only if 
it can be represented as a finite sum, where each summand can be represented as 
$\frac{A_1 B_1 + \ldots + A_n B_n}{p}$,
for some nonnegative integer $n$ 
and some $A_1$, \ldots, $A_n \in \A$, $B_1$, \ldots, $B_n \in \B$, $p \in \rat(a, b)$, 
with an additional condition $p \neq 0$. 
The last condition is necessary because otherwise the denominator would be equal to zero 
and the fraction would not make sense.
\end{ex} 

\begin{ex}
Laurent series
$f := \sum_{n = 0}^{+\infty} \sum_{m = 0}^{+\infty} (a^{2^n} b^{2^m - 10} + a^{2^n + 10} b^{2^m + 15})$
match $\dfrac{\sum \A\B}{\poly(a, b)}$,
because $f = \dfrac{A_1 B_1 + A_2 B_2}{b^{10}}$, where $b^{10} \in \poly(a, b)$,
$A_1 = \sum_{n=0}^{+\infty} a^{2^n} \in \A$, $B_1 = \sum_{m=0}^{+\infty} b^{2^m} \in \B$,
$A_2 = \sum_{n=0}^{+\infty} a^{2^n + 10} \in \A$, $B_2 = \sum_{m=0}^{+\infty} b^{2^m + 25} \in \B$.
\end{ex}

\begin{defn} Algebraic expressions $F$ and $G$ \emph{are equivalent} (denoted by $F = G$), 
if they define the same subset of the whole field $\F((a, b))$ of Laurent series in variables
$a$ and $b$.
\end{defn}

Some equivalencies follow directly from definitions and properties of classes 
$\poly(a, b)$, $\rat(a, b)$, $\A$ and $\B$.
For example, $\sum \sum \A \B = \sum \A \B$, 
aforementioned $\rat(a, b) = \dfrac{\poly(a, b)}{\poly(a, b)}$, $\sum \A \A = \A$ and 
$\dfrac{\sum \A \B}{\sum \A \B} = \sum \dfrac{\A\B}{\sum \A\B}$. 

On the other hand, some equivalences are not so trivial, like the equivalence 
$\dfrac{\sum \A \B}{\sum \A \B} = \dfrac{\sum \A \B}{\poly(a, b)}$, which I shall establish later.

\subsection{Switching to the algebraic track}

The purpose of this section is to prove the following intermediate result:

\begin{lm}\label{Translation} Assume that language $K \subseteq a^*  b^*$ is described by a GF(2)-grammar.
Then, the corresponding power series $\Dual(K)$ match $\dfrac{\sum \A\B}{\sum \A\B}$.
\end{lm}

\begin{proof} Without loss of generality, the GF(2)-grammar 
that describes $K$ is in the Chomsky normal form~\cite[Theorem 5]{gf2}.  
Moreover, let us assume that $K$ does not contain the empty string.

The language $a^* b^*$ is recognized by the following 
incomplete deterministic finite automaton $M$:
$M$ has two states $q_a$ and $q_b$, both accepting, and its transition function is
$\delta(q_a, a) = q_a, \delta(q_a, b) = q_b, \delta(q_b, b) = q_b$.

Let us formally intersect the GF(2)-grammar $G$ with a regular language 
$a^* b^*$, recognized by the automaton $M$, 
using the construction of Bar-Hillel et al.~\cite{BarhillelPerlesShamir}
(the construction of the intersection of an ordinary grammar with a regular expression by Bar-Hillel 
et al.~\cite{BarhillelPerlesShamir} can be easily adapted to the case of GF(2)-grammars~\cite [Section 6]{grammars-gf2}).
The language described by the GF(2)-grammar will not change, 
because it was already a subset of $a^* b^*$.

The grammar will change considerably, however.
Every nonterminal $C$ of the original grammar 
splits into three nonterminals: $C_{a \to a}, C_{a \to b}, C_{b \to b}$. 
Also a new starting nonterminal $S'$ apears.

Every ``normal'' rule $C \to DE$ splits into four rules:
$C_{a \to a} \to D_{a \to a} E_{a \to a}$, $C_{a \to b} \to D_{a \to a} E_{a \to b}$, 
$C_{a \to b} \to D_{a \to b} E_{b \to b}$ and $C_{b \to b} \to D_{b \to b} E_{b \to b}$. 

The following happens with ``final'' rules: $C \to b$ turns into two rules $C_{a \to b} \to b$ and $C_{b \to b} \to b$, and $C \to a$ turns into one rule $C_{a \to a} \to a$. 
Finally, two more rules appear: $S' \to S_{a \to a}$ and $S' \to S_{a \to b}$.

What do the nonterminals of the new GF(2)-grammar correspond to? The state $C_{a \to a}$ corresponds to the strings $w \in \{a, b\}^*$ that are derived from the nonterminal $C$ 
of the original GF(2)-grammar and make $M$ go from the state $q_a$ to itself.
Formally speaking, $w \in L(C_{a \to a})$ if and only if $w \in L(C)$ and $\delta(q_a, w) = q_a$.
Similarly, $w \in L(C_{a \to b})$ if and only if $w \in L(C)$ and $\delta(q_a, w) = q_b$.
Finally, $w \in L(C_{b \to b})$ if and only if $w \in L(C)$ and $\delta(q_b, w) = q_b$.

By looking more closely on the transitions of $M$, we can see that
$\delta(q_a, w) = q_a$ if and only if $w$ consists only of letters $a$, in other words, 
if and only if $w \in a^*$.
Similarly, $\delta(q_b, w) = q_b$ if and only if $w \in b^*$, and 
$\delta(q_a, w) = q_b$ if and only if $w \in a^* b^+$. 

Every language $L(C_{a \to a})$ is a 2-automatic language over a unary alphabet $\{a\}$. 
Indeed, every parse tree of $C_{a \to a}$ contains only nonterminals of type $a \to a$. 
Therefore, only character $a$ can occur as a \emph{terminal} in a parse tree of $C_{a \to a}$. 
So, $L(C_{a \to a})$ is described by some GF(2)-grammar over an alphabet $\{a\}$, 
and is therefore 2-automatic.
Similarly, all languages $L(C_{b \to b})$ are 2-automatic over the alphabet $\{b\}$. 
Then, by Christol's theorem, $\Dual(L(C_{a \to a})) \in \A$ and $\Dual(L(C_{b \to b})) \in \B$. 

How do the languages $L(C_{a \to b})$ look like? 
Let us look at the rules $C_{a \to b} \to D_{a \to a} E_{a \to b}$
and $C_{a \to b} \to D_{a \to b} E_{b \to b}$. These rules can be interpreted in the following way: when starting a parse from nonterminal $C_{a \to b}$, 
we can append a language from $\A$ from the left and go to $E_{a \to b}$
or append a language from $\B$ from the right and go to $D_{a \to b}$.

What can we say about $K$? 
By definition, $K = L(S) = L(S') = L(S_{a \to a}) \triangle L(S_{a \to b})$.
We can forget about the language $L(S_{a \to a})$: it is from the class $\A$, and $L(S_{a \to b})$
is from much more complicated class, that will ``absorb'' $\A$ in the end.

The languages $L(C_{a \to b})$ for each nonterminal $C_{a \to b}$ of the new grammar
satisfy the following system of language equations:

\begin{equation}\label{LangSystem_old} L(C_{a \to b}) = \fin(C_{a \to b}) \triangle \bigtriangleup_{C \to DE} 
(L(D_{a \to a}) \odot L(E_{a \to b})) \triangle (L(D_{a \to b}) \odot L(E_{b \to b}))
\end{equation}
	
Here, the summation happens over all rules $C \to DE$ for each nonterminal $C$ of the original
grammar, and $\fin(C_{a \to b})$ is either $\{b\}$ or $\varnothing$, depending on whether or not
there is a rule $C_{a \to b} \to b$ in the new grammar.

Look more closely at the system~\eqref{LangSystem_old}.
In all GF(2)-concatenations that appear in its right-hand side
either the first language is a subset of $a^*$, or the second
language is a subset of $b^*$. Hence, we can apply the Lemma~\ref{comm-concat}.

Denote $\Dual(L(C_{a \to b}))$ by $\Center(C)$, $\Dual(L(C_{a \to a}))$ by $\Left(C)$,
$\Dual(L(C_{b \to b}))$ by $\Right(C)$ and $\Dual(\fin(C_{a \to b}))$ by $\final(C)$ for brevity.

Applying $\Dual$ to the both sides of~\eqref{LangSystem_old} gives us the following
system of equations over formal power series:

\begin{equation}\label{SeriesSystem_old} \Center(C) = \final(C) + \sum\limits_{C \to DE}
\Left(D) \Center(E) + \Center(D) \Right(E)
\end{equation}

Let us look at this system as a system of $\F[[a, b]]$-linear equations over variables $\Center(C) = \Dual(L(C_{a \to b}))$ for
every nonterminal $C$ of the original GF(2)-grammar.

We will consider $\final(C)$, $\Left(C)$ and $\Right(C)$ to be the coefficients of the system. 
While we do not know their \emph{exact} values, the following is known: $\final(C)$ is $0$ or $b$, $\Left(C) \in \A$ as a formal power series that
corresponds to a 2-automatic language over an alphabet $\{a\}$ and, similarly, $\Right(C) \in \B$.

Denote the number of nonterminals in the original GF(2)-grammar by $n$,
(so there are $n$ nonterminals of type $a \to b$ in the new GF(2)-grammar),
a column vector of values $\Center(C)$ by $x$ and a column vector of values 
$\final(C)$ in the same order by $f$.  
Let us fix the numeration of nonterminals $C$ of the old GF(2)-grammar.
After that, we can use them as the ``indices'' of rows and columns of matrices.

Let $I$ be an identity matrix of dimension $n \times n$, $A$ be a $n \times n$ matrix
with the sum of $\Left(D)$ over all rules $C \to DE$ of the original grammar standing on the intersection of $C$-th row and $E$-th column:
\begin{equation}\label{Adef_old}
 A_{C, E} := \sum\limits_{C \to DE} \Left(D) 
\end{equation}

Similarly, let $B$ be a $n \times n$ matrix with
\begin{equation}\label{Bdef_old}
 B_{C, D} := \sum\limits_{C \to DE} \Right(E)
\end{equation}

Then the equation system~\eqref{SeriesSystem_old} can be rewritten as $x = f + (A + B) x$ in the
matrix form. In other words, $(A + B + I)x = f$.

We have already proven earlier that $\Center(C)$ is a solution of this system. Our plan is to
prove that there is exactly one solution to this system and express it in some form.
Then, in particular, we will find some expression for $\Dual(L(S_{a \to b})) = \Center(S)$.

This system has exactly one solution if and only if $\det(A + B + I) \neq 0$. If $\det (A + B + I) \neq 0$,
then, by Cramer's formula, every entry of the solution, including $\Center(S)$ can be written as
\begin{equation*}
\dfrac{\det(\mbox{$A + B + I$, but with one of the columns replaced by $f$})}{\det(A + B + I)}
\end{equation*}

It remains to establish three things: that $\det(A + B + I)$ matches $\sum \A\B$, that
$\det(\mbox{$A + B + I$, but with one of the columns replaced by $f$})$ matches $\sum \A\B$, independently of the replaced column and that $\det (A + B  + I) \neq 0$.

Let us prove the first two statements at the same time. 
Every entry of $A + B + I$ matches $\A + \B$
because of the equations~\eqref{Adef_old}--\eqref{Bdef_old}.
Indeed, every entry of $A$ matches $\sum \A = \A$,
every entry of $B$ matches $\B$,
and entries of $I$ are ones and zeroes that lie in both $\A$ and $\B$.
This property will not disappear, if you replace every column of the matrix
by $f$: all entries of $f$ are equal to $0$ or $b$, so they match $\B$, let alone $\A + \B$.

Now, let us prove that the determinant of the matrix with every entry matching $\A + \B$
matches $\sum \A\B$. Indeed, by expressing the determinant through the explicit formula
with $n!$ summands, we get that the determinant matches 
\begin{equation*}
\sum \underbrace{(\A + \B) \cdot \ldots \cdot (\A + \B)}_{n \text{ times}}.
\end{equation*}
By expanding the brackets and using the fact that $\A \A = \A$ and $\B \B = \B$, we see
that the determinant match $\sum \A\B$.

It remains to prove that $\det(A + B + I) \neq 0$. Let us prove a stronger statement: that the
power series $\det(A + B + I) \in \F[[a, b]]$ is invertible, that is, its coefficient at $a^0 b^0$
is equal to $1$.

Notice that finite product of power series is invertible if and 
only if each factor is invertible. Also a finite sum of power series 
with exactly one invertible summand  is invertible. 

Because the new GF(2)-grammar is also in Chomsky's normal form,
all languages $L(C_{a \to a})$ and $L(C_{b \to b})$ do not contain the empty string.
Therefore, all series $\Left(C)$ and $\Right(C)$ are invertibe. Therefore, by
equations~\eqref{Adef_old}--\eqref{Bdef_old}, all entries of $A + B$ are invertible.
It follows that exactly the diagonal entries of $A + B + I$ are invertible: 
they are obtained by adding one to invertible series, and other entries of $A + B + I$
coincide with the same entries of $A + B$.

Let us use the formula for $\det(A + B + I)$ with $n!$ summands again. Exactly
one summand is invertible: the one that corresponds to the identity permutation.
Indeed, all other summands have at least one nondiagonal, therefore, non-invertible,
element. And the summand that corresponds to the identity permutation is 
a product of the diagonal entries of $A + B + I$. Hence, said summand is invertible as power series.

We have just proved that $\det(A + B + I)$ is invertible. In particular, $\det(A + B + I) \neq 0$.

Now we can use the Cramer's formula and conclude that  $\Dual(L(S_{a \to b}))$ match $\dfrac{\sum \A\B}{\sum \A\B}$.
Then $\Dual(K) = \Dual(L(S')) = \Dual(L(S_{a \to a}) + \Dual(L(S_{a \to b}))$ match
$\A + \dfrac{\sum \A\B}{\sum \A\B} = \dfrac{\sum \A\B}{\sum \A\B}$.
The last equivalence holds, because we can
find a common denominator for the summands
and obtain $\A \sum \A\B + \sum \A\B = \sum \A \B$ in the numerator. 

\end{proof}

\subsection{Algebraic manipulations}

The purpose of this section is to prove
 the following theorem:

\begin{thm}\label{Main_ab_old} If $L \subseteq a^* b^*$ is described by a GF(2)-grammar. Then,
the corresponding power series $\Dual(L)$ match $\dfrac{\sum \A\B}{\poly(a, b)}$.
\end{thm}

In the previous section, we have already moved to this goal, by dealing with the language-theoretic details.
Now, we want to use some algebraic manipulations. Theorem~\ref{Main_ab_old} 
would follow from the Lemma~\ref{Translation} and the following lemma:

\begin{lm}\label{Manipulation} Algebraic expressions $\dfrac{\sum \A\B}{\sum \A\B}$ and
$\dfrac{\sum \A\B}{\poly(a, b)}$ are equivalent.
\end{lm}

\begin{rem} It is immediately apparent that Theorem~\ref{Main_ab_old} is exactly Theorem~\ref{Main_ab}, but stated in terms of algebraic expressions. What can be more difficult to see is the fact that Lemma~\ref{Manipulation} is also, essentially, a restatement of Lemma~\ref{lemma_rab}. Indeed, representing $\dfrac{\sum \A\B}{\sum \A\B}$ as $\dfrac{\sum \A\B}{\poly(a, b)}$ is exactly the same as proving that a ratio of two elements of $R_{a, b}$ is still an element of $R_{a, b}$. And the proof of Lemma~\ref{Manipulation} is, essentially, a roundabout way to prove the following well-known fact: if we adjoin an arbitrary number of algebraic elements to a field, the result will still be a field. 
\end{rem} 

\begin{proof}[Proof of Lemma~\ref{Manipulation}]
It is evident that the algebraic expression $\dfrac{\sum \A\B}{\sum \A\B}$ is not weaker than 
$\dfrac{\sum \A\B}{\poly(a, b)}$, because $\sum \A \B$ is not weaker than $\poly(a, b)$. 
Indeed, every polynomial is a finite sum of monomials of type $a^n b^m$, and
each such monomial match $\A \B$, because $a^n \in \A$
and $b^m \in \B$. 

The converse implication, namely that
$\dfrac{\sum \A\B}{\sum \A\B}$ is not \emph{stronger} than $\dfrac{\sum \A\B}{\poly(a, b)}$,
is more interesting.
Suppose that some formal power series $f$ match $\dfrac{\sum \A\B}{\sum \A\B}$. Then it also match $\sum \dfrac{\A\B}{\sum \A\B}$. Let us show that each of the summands match
$\dfrac{\sum \A \B}{\poly(a, b)}$, then the whole sum match
$\sum \dfrac{\sum \A \B}{\poly(a, b)} = \dfrac{\sum \A \B}{\poly(a, b)}$, as intended.

Indeed, suppose that we have some Laurent series matching the expression 
$\dfrac{\A\B}{\sum \A\B}$. Then, by definition
of ``matching algebraic expression'', these series are of the type 
$\frac{A_0 B_0}{A_1 B_1 + A_2 B_2 + \ldots + A_n B_n}$, where $n$ is a positive integer, and 
$A_i \in \A, B_i \in \B$ for every $i$ from $0$ to $n$ inclusive. Moreover, this expression makes sense, 
meaning that $A_1 B_1 + \ldots + A_n B_n \neq 0$.

We still have not used that $\A$ and $\B$ are exactly the sets of algebraic power series, and not just some
subsets of $\F[[a]]$ and $\F[[b]]$ that are closed under addition. Let us use that.

More exactly, we want to get rid of difficult expression in the numerator by rewriting 
$\frac{1}{A_1 B_1 + \ldots + A_n B_n}$, that is, $(A_1 B_1 + \ldots + A_n B_n)^{-1}$,
as a finite $\rat(a, b)$-linear combination of \emph{nonnegative} powers of $A_1 B_1 + \ldots + A_n B_n$.

The least painful way to do so is to find a nontrivial $\rat(a, b)$-linear dependence between nonnegative
powers of $A_1 B_1 + \ldots + A_n B_n$ and then get the required expression from it.
It still is not very easy, see below for details.

Because every $A_i$ is an algebraic power series in variable $a$ over the ring $\F[a]$, it also is
an algebraic power series in variables $a$ and $b$ over the field $\rat(a, b)$: the same polynomial
equation will suffice to show that.

We will need a few technical lemmas:

\begin{lm} Suppose that Laurent series $f \in \F((a,b))$ is a solution to a polynomial equation of degree $d$ with coefficients
from $\rat(a, b)$. Then, for every $m \geqslant d$, the power series $f^m$ can be represented as a $\rat(a, b)$-linear
combination of $f^{m - 1}, f^{m - 2}, \ldots, f^{m - d}$.
\end{lm}

\begin{proof} Indeed, by conditions of the lemma, $\sum_{i=0}^d p_i f^i = 0$ for some $p_i \in \rat(a, b)$.
Moreover, $p_d \neq 0$, because the degree of the equation is exactly $d$. Divide both sides by $p_d$
and move $f^d$ to the right-hand side: $\sum_{i=0}^{d - 1} \dfrac{p_i}{p_d} f^i = f^d$.
Multiply both sides by $f^{m - d}$: 
$\sum_{j=m-d}^{m-1} \frac{p_{j - (m - d)}}{p_d} f^j = f^m$, exactly a representation of $f^m$
as a $\rat(a, b)$-linear combination of $f^{m-1}, f^{m-2}, \ldots, f^{m-d}$.
\end{proof}
 
\begin{lm} Suppose that Laurent series $f \in \F((a,b))$ is a root of a polynomial equation of degree $d$ with coefficients
from $\rat(a, b)$. Then, for every $m \geqslant 0$, $f^m$ can be represented as a $\rat(a, b)$-linear combination
of $f^{d - 1}, f^{d-2}, \ldots, f^0$. In other words, all nonnegative powers of $f$ are in the $\rat(a,b)$-linear space
generated by $f^0, f^1, \ldots, f^{d-1}$.
\end{lm}
\begin{proof} Induction over $m$. Denote the $\rat(a,b)$-linear space, generated by $f^0, f^1, \ldots, f^{d-1}$ by $L$.
The statement is trivially true for $m < d$, because $f^m$ is one of generators of $L$.

Now, suppose that we want to prove the statement of the lemma for some $m \geqslant d$. By induction hypothesis,
$f^{m-1}, f^{m-2}, \ldots, f^{m-d}$ all lie in $L$. By the previous lemma, $f^m$ can be represented as a 
$\rat(a,b)$-linear combination of $f^{m-1}, f^{m-2}, \ldots, f^{m-d}$. Therefore, $f^m$ lies in $L$ as
a finite $\rat(a,b)$-linear combination of elements of $L$.
\end{proof}

\begin{lm} 
There is no infinite subset of monomials in variables $A_i$ and $B_i$, that is 
 linearly independent over $\rat(a,b)$.
In other words, the $\rat(a,b)$-linear space generated by all values of polynomials in variables
$A_1, A_2, \ldots, A_n$ and $B_1, B_2, \ldots, B_n$ with coefficients from $\rat(a,b)$ is finite-dimensional.
\end{lm}
\begin{proof} Because $A_i \in \A$ for every $i$ from $1$ to $n$ inclusive, there are some $\ell_i$ such that
$A_i$ is a root of degree-$d$ polynomial equation with coefficients from $\rat(a,b)$. Similarly, denote
by $r_i$ the degrees of polynomial equations for $B_i$.

Let us try to represent expression $A_1^{j_1} \ldots A_n^{j_n} \cdot B_1^{k_1} \ldots B_n^{k_n}$ 
for some nonnegative $j_s$ and $k_s$ as a $\rat(a,b)$-linear combination of similar expressions
with \emph{small} degrees.

Indeed, by previous lemma, every $A_s^{j_s}$ is a $\rat(a,b)$-linear combination of $A_s^0, A_s^1, \ldots, A_s^{\ell_s - 1}$.
Similarly, every $B_s^{k_s}$ is a $\rat(a,b)$-linear combination of $B_s^0, B_s^1, \ldots, B_s^{r_s - 1}$.
Represent $A_1^{j_1} \ldots A_n^{j_n} \cdot B_1^{k_1} \ldots B_n^{k_n}$ as a product of such linear combination
and expand all brackets. The result is some $\rat(a,b)$-linear combination of expressions $A_1^{x_1} \ldots A_n^{x_n} \cdot B_1^{y_1} \ldots B_n^{y_n}$, but with $0 \leqslant x_s < \ell_s$ and $0 \leqslant y_s < r_s$.

Let $L$ be the $\rat(a,b)$-linear space generated by all products of type
$A_1^{x_1} A_2^{x_2} \ldots A_n^{x_n} \cdot B_1^{y_1} \ldots B_n^{y_n}$, where $0 \leqslant x_s < \ell_s$ and 
$0 \leqslant y_s < r_s$ for all $s$ from $1$ to $n$ inclusive. This space is generated by $\ell_1 \ell_2 \ldots \ell_n \cdot
r_1 r_2 \ldots r_n$ elements and therefore is finite-dimensional.

We have already established that every monomial $A_1^{j_1} \ldots A_n^{j_n} \cdot B_1^{k_1} \ldots B_n^{k_n}$ is
a $\rat(a,b)$-linear combination of elements of $L$ (moreover, exactly the elements that were $L$'s generators),
therefore it lies in $L$. Then, every polynomial expression in variables $A_s$ and $B_s$ lies in $L$, as a linear combination of monomials that lie in $L$.
\end{proof}

Because the space of \emph{all} polynomial expression of $A_i$ and $B_i$ is finite-dimensional, the space
generated by nonnegative powers of $A_1 B_1 + \ldots + A_n B_n$ also is. Therefore, there exists a nontrivial
$\rat(a,b)$-linear dependence between nonnegative powers of $A_1 B_1 + \ldots + A_n B_n$. In other words,
there is some nonnegative integer $d$ and rational functions $p_0, p_1, \ldots, p_d \in \rat(a,b)$, not all equal to zero,
such that $\sum_{i=0}^d p_i (A_1 B_1 + \ldots + A_n B_n)^i = 0$. Let us find an expression of
$(A_1 B_1 + \ldots + A_n B_n)^{-1}$ through nonnegative powers of $A_1 B_1 + \ldots + A_n B_n$ with that
knowledge.

Indeed, let us take the smallest such $j$ that $p_j \neq 0$. It exists, because not all $p_i$ are equal to zero.
Then, our equation can be rewritten as $\sum_{i=j}^d p_i (A_1 B_1 + \ldots + A_n B_n)^i = 0$, because
$p_0 = p_1 = \ldots = p_{j-1} = 0$ anyways. By dividing both sides by $p_j (A_1 B_1 + \ldots + A_n B_n)^{j + 1} \neq 0$,
we obtain 
\begin{equation*}
\sum_{i = j}^d  \dfrac{p_i}{p_j} (A_1 B_1 + \ldots + A_n B_n)^{i - j - 1} = 0.
\end{equation*}
All powers of 
$A_1 B_1 + \ldots + A_n B_n$ from $(-1)$-st to $(d-j-1)$-st are here with some coefficients, the coefficient 
before $(-1)$-st power is $p_j / p_j = 1$. By moving all powers, except $(-1)$-st to the right-hand side, we obtain
\begin{equation*}
(A_1 B_1 + \ldots + A_n B_n)^{-1} = \sum_{i=j+1}^d \dfrac{p_i}{p_j} (A_1 B_1 + \ldots + A_n B_n)^{i - j - 1} =  \sum_{i=0}^{d-j-1} \dfrac{p_{i+j+1}}{p_j} (A_1 B_1 + \ldots + A_n B_n)^{i}.
\end{equation*}

Therefore,  $(A_1 B_1 + \ldots + A_n B_n)^{-1}$ match $\sum \rat(a,b) \A \B$ (to understand that, 
expand all brackets in the right-hand side). Therefore, $\frac{A_0 B_0}{A_1 B_1 + \ldots + A_n B_n}$ match $\A \B \cdot \sum \rat(a,b) \A \B = \sum \rat(a, b) \A \B$. We are almost done!

\begin{rem} 
Generally speaking, $\rat(a,b)$ cannot be split into two parts with the first being ``absorbed'' by $\A$ and
the second being ``absorbed'' by $\B$. Keep the following example in the head: $1 + ab$. It is not hard
to prove that $1 + ab$ is not a product of a factor depending only on $a$ and a factor depending only on $b$.
\end{rem}

As we understood earlier, every Laurent series matching $\dfrac{\A\B}{\sum \A\B}$
also match
$\sum \rat(a,b) \A\B$. Then all Laurent series matching $\dfrac{\sum \A\B}{\sum \A\B}$
also match $\sum \sum \rat(a,b) \A\B = \sum\rat(a, b) \A \B$. 
Finally, by adding the fractions up,
every Laurent series matching $\sum \rat(a,b) \A \B$ match
$\dfrac{\sum \poly(a, b) \A \B}{\poly(a, b)} = \dfrac{\sum \A\B}{\poly(a, b)}$.
\end{proof}

\section{Subsets of \texorpdfstring{$a^*b^*c^*$}{a*b*c*} through Algebraic Expressions}\label{alg_expr_abc}

The language $\set{a^n b^n c^n}{n \geqslant 0}$ is, probably, the most famous 
example of a simple language that is not described by any ordinary grammar. 
It is reasonable to assume that it is not described by a GF(2)-grammar as well.
Let us prove that.

We will do more than that and will actually establish some property that
all GF(2)-grammatical subsets of $a^* b^* c^*$ have, but 
$\set{a^n b^n c^n}{n \geqslant 0}$ does not. 
Most steps of the proof will be analogous to the two-letter case.

There is a natural one-to-one correspondence between subsets of $a^* b^* c^*$ 
and formal power series in variables $a, b$ and $c$ over field $\F$. Indeed, for every
set $S \subseteq \N_0^3$, we can identify the language $\set{a^n b^m c^k}{(n, m, k) \in S}
\subseteq a^* b^* c^*$ with the formal power series 
$\sum\limits_{(n, m, k) \in S} a^n b^m c^k$. Denote this corres`pondence 
by $\Dual \colon 2^{a^* b^* c^*} \to \F[[a, b, c]]$. Then, $\Dual(L \triangle K)
= \Dual(L) + \Dual(K)$. In other words, the symmetric difference of languages
corresponds to the sum of formal power series.

Similarly to the Lemma~\ref{comm-concat}, $\Dual(K \odot L) = \Dual(K) \cdot \Dual(L)$
in the following important special cases: when $K$ is a subset of $a^*$, 
when $K$ is a subset of $a^* b^*$ and $L$ is a subset of $b^* c^*$, and, finally, when $L$
is a subset of $c^*$.	 Indeed, in each of these three cases, characters ``are in the correct order'':
if $u \in K$ and $v \in L$, then $uv \in a^* b^* c^*$.

However, we cannot insert character $b$ in the middle of the string: if $K$ is a subset 
of $b^*$ and $L$ is a subset of $a^* b^* c^*$, then $K \odot L$ and $K \odot_\mathrm{comm} L$ do not have to coincide, because $K \odot L$ does not even have to be a subset of $a^* b^* c^*$. 

The ``work plan'' will remain the same as in the previous section: we will switch to 
algebraic track first and then we simplify the expression obtained.

An attentive reader may ask two questions:
\begin{enumerate}
\item Why is it logical to expect that the language $\set{a^n b^n c^n}{n \geqslant 0}$
is not described by a GF(2)-grammar, but a similar language $\set{a^n b^n}{n \geqslant 0}$
is?
\item Why will the proof work out for $\set{a^n b^n c^n}{n \geqslant 0}$, but not
for a regular language $\set{(abc)^n}{n \geqslant 0}$, despite these languages
having the same ``commutative image''?
\end{enumerate}

They can be answered in the following way:

\begin{enumerate}
\item Simply speaking, the reason is the same as for the ordinary grammars. On a
intuitive level, both ordinary grammars and GF(2)-grammars permit a natural
way to ``capture'' the events that happen with any two letters in subsets of $a^* b^* c^*$,
but not all three letters at the same time. A rigourous result that corresponds to this 
intuitive limitation of ordinary grammars was proven by Ginsburg and Spanier~\cite[Theorem 2.1]{bounded-original}. Theorem~\ref{Main_abc} is an analogue for GF(2)-grammars.

\item This argument only implies that any proof that relies \emph{solely} on commutative
images is going to fail. The real proof is more subtle. For example, it will also use the fact
that $\set{a^n b^n c^n}{n \geqslant 0}$ is a subset of $a^* b^* c^*$.

While the proof uses commutative images, it uses them very carefully, always making sure
that the letters ``appear in the correct order''. In particular, we will never consider
GF(2)-concatenations $K \odot L$, where $K$ is a subset of $b^*$
and $L$ is an arbitrary subset of $a^* b^* c^*$, in the proof, 
because in this case $K \odot L$ is not a subset of $a^* b^* c^*$.

Avoiding this situation is impossible for language $\set{(abc)^n}{n \geqslant 0}$, because
in the string $abcabc$ from this language the letters ``appear in the wrong order''.
\end{enumerate}

Denote the set of algebraic power series of variable $c$ by $\C$, the set of polynomials
in variables $a$ and $c$ by $\poly(a, c)$, et cetera.
The definition of an algebraic expression stays the same for the most part, but now the new symbols
$\C$, $\poly(a, c)$, $\poly(b, c)$, $\poly(a, b, c)$, $\rat(a, c)$, $\rat(b, c)$ and $\rat(a, b, c)$
may appear alongside the old symbols $\A$, $\B$ and $\poly(a, b)$.

\subsection{Switching to the algebraic track}

Our goal for this subsection is to establish the following lemma:
\begin{lm}\label{Translation_abc} Suppose that $K \subseteq a^* b^* c^*$ is described by a GF(2)-grammar.
Then the corresponding formal power series $\Dual(K)$ match algebraic expression
$\dfrac{\sum \A\B\C}{\poly(a, b) \poly(b, c) \cdot \sum \A\C}$.
\end{lm}
\begin{proof} The proof is mostly the same as the proof of Lemma~\ref{Translation}.

Without loss of generality, GF(2)-grammar $G$ that describes $K$ is in Chomsky's normal form.
Also we can assume that $K$ does not contain the empty string.

The language $a^* b^* c^*$ is accepted by the following incomplete deterministic finite
automaton $M$. Firstly, $M$ has three states $q_a$, $q_b$ and $q_c$, all accepting.
Secondly, its transition function $\delta$ is defined as $\delta(q_a, a) = q_a, \delta(q_a, b) = q_b,
\delta(q_a, c) = q_c, \delta(q_b, b) = q_b, \delta(q_b, c) = q_c, \delta(q_c, c) = q_c$.

Intersect the GF(2)-grammar $G$ formally with regular language $a^* b^* c^*$, recognized by $M$. Because $L(G) = K$ was a subset of $a^* b^* c^*$ anyway, the described language will not
change. Each nonterminal $C$ of the original grammar will split into six nonterminals 
$C_{a \to a}$, $C_{a \to b}$, $C_{a \to c}$, $C_{b \to b}$, $C_{b \to c}$, $C_{c \to c}$.
Also, a new starting nonterminal $S'$ will appear.	 

Every ``normal'' rule $C \to DE$
will split into rules $C_{a \to a} \to D_{a \to a} E_{a \to a}$, $C_{a \to b} \to 
D_{a \to a} E_{a \to b}$, $C_{a \to b} \to D_{a \to b} E_{b \to b}$, $C_{a \to c}
\to D_{a \to a} E_{a \to c}$, $C_{a \to c} \to D_{a \to b} E_{b \to c}$, $C_{a \to c}
\to D_{a \to c} E_{c \to c}$, $C_{b \to c} \to D_{b \to b} E_{b \to c}$, $C_{b \to c}
\to D_{b \to c} E_{c \to c}$, and $C_{c, c} \to D_{c \to c} E_{c \to c}$.
Less horrifying than it looks, because most of these rules will not be interesting 
to us in the slightest.

A ``final'' rule $C \to a$ will turn into a rule $C_{a \to a} \to a$. Similarly, rule $C \to b$
will split into two rules $C_{a \to b} \to b$ and $C_{b \to b} \to b$, and rule $C \to c$
will split into three rules $C_{a \to c} \to c$, $C_{b \to c} \to c$ and $C_{c \to c} \to c$. 

Finally, three new rules will appear: $S' \to S_{a \to a}$, $S' \to S_{b \to b}$, $S' \to S_{c \to c}$.

The nonterminal $C_{x \to y}$ of the new GF(2)-grammar, where $x, y \in \{a, b, c\}$, 
corresponds to exactly such strings from $L(C)$ that move the automaton $M$ from the state
$q_x$ to the state $q_y$.

By looking more closely at the transitions of the automaton $M$, we can see that any string
that makes $M$ go from $q_a$ to $q_c$ is from $a^* b^* c^+$, any string that makes
$M$ go from $q_b$ to $q_c$ is from $b^* c^+$, et cetera. In particular, in all new
``normal'' rules GF(2)-concatenations happen ``in the correct order''.

As already mentioned, most of the new rules are not interesting, because we already
know, how the languages $L(C_{a \to a})$, $L(C_{a \to b})$, $L(C_{b \to b})$,
$L(C_{b \to c})$ and $L(C_{c \to c})$ look like. More specifically, the corresponding
formal power series match algeebraic expressions $\A$, $\dfrac{\sum \A\B}{\poly(a, b)}$,
$\B$, $\dfrac{\sum \B\C}{\poly(b, c)}$ and $\C$ respectively.

Therefore, we are only interested in nonterminals of the type $a \to c$. The rule $X \to YZ$
of the original GF(2)-grammar $G$ produces three rules for $X_{a \to c}$:
$X_{a \to c} \to Y_{a \to c} Z_{c \to c}$, $X_{a \to c} \to Y_{a \to b} Z_{b \to c}$
and $X_{a \to c} \to Y_{a \to a} Z_{a \to c}$. The first and last rule relate $L(X_{a \to c})$
to other nonterminals of type $a \to c$, and the second rule just outright tells us that
we can replace $X_{a \to c}$ with a language matching $\dfrac{\sum \A\B}{\poly(a, b)} \cdot
\dfrac{\sum \B\C}{\poly(b, c)}$. Finally, there may be a final rule $X_{a \to c} \to c$ for nonterminal $X_{a \to c}$.

We can conclude that the languages $L(C_{a \to c})$ satisfy the following system of 
language equations.
\begin{equation}\label{Lang_abc} L(C_{a \to c}) = \fin(C_{a \to c}) \triangle
\bigtriangleup\limits_{C \to DE} (L(D_{a \to a}) \odot L(E_{a \to c})) \triangle
(L(D_{a \to c}) \odot L(E_{c \to c}))
\end{equation}

Here, the summation happens over all rules $C \to DE$ for the nonterminal $C$ of the 
original GF(2)-grammar, and $\fin(C_{a \to c})$ is defined as:
\begin{equation}\label{end_abc_def} \fin(C_{a \to c}) = 
(\{c\} \text { or } \varnothing) \triangle \bigtriangleup\limits_{C \to DE} L(D_{a \to b}) \odot L(E_{b \to c})
\end{equation}

Here, the first ``summand'' depends on whether or not there is a rule $C_{a \to c} \to c$ in the new
GF(2)-grammar. 

Consider the equations from System~\eqref{Lang_abc} more closely. For all GF(2)-concatenations
that appear in their right-hand sides, either the first factor is a subset of $a^*$, or the 
second is a subset of $c^*$. Therefore, we can replace all GF(2)-concatenations here satisfy the conditions of Lemma~\ref{comm-concat}
\begin{equation}\label{Lang_abc_comm}
L(C_{a \to c}) = \fin(C_{a \to c}) \triangle
\bigtriangleup\limits_{C \to DE} (L(D_{a \to a}) \odot_{\mathrm{comm}} L(E_{a \to c})) \triangle 
(L(D_{a \to c}) \odot_{\mathrm{comm}} L(E_{c \to c}))
\end{equation}

Denote $\Dual(L(C_{a \to c}))$ by $\Center(C)$, $\Dual(L(C_{a \to a}))$ by 
$\Left(C)$, $\Dual(L(C_{c \to c}))$ by $\Right(C)$ and $\Dual(\fin(C_{a \to c}))$ by
$\final(C)$. By applying the correspondence $\Dual$ to the both sides of each equation of the
System~\eqref{Lang_abc_comm}, 
\begin{equation}\label{series_abc}
\Center(C) = \final(C) + \sum\limits_{C \to DE} \Left(D) \Center(E) + \Center(D) \Right(E)
\end{equation}

This system of equation can be interpeted as a system $\F[[a, b, c]]$-linear equations
over the variables $\Center(C) = \Dual(L(C_{a \to c}))$ for every nonterminal $C$
of the original GF(2)-grammar.

We will consider $\final(C)$, $\Left(C)$ and $\Right(C)$ to be the coefficients of said system.
While we do not know their exact values, we know that $\final(C)$ match the expression
$\sum \left( \dfrac{\sum \A\B}{\poly(a, b)} \right) \cdot 
\left( \dfrac{\sum \B\C}{\poly(b, c)} \right) = \dfrac{\sum \A\B\C}{\poly(a, b) \poly(b,c)}$  by 
formula~\eqref{end_abc_def} and Theorem~\ref{Main_ab_old},  $\Left(C)$ is in $\A$, because it 
corresponds to a 2-automatic language over an alphabet $\{a\}$ and, similarly,
$\Right(C)$ is in $\C$.

Let us say that the original GF(2)-grammar has $n$ nonterminals. Then, the new
GF(2)-grammar has $n$ nonterminals of type $a \to c$. Denote the column-vector
of values $\Center(C)$ by $x$, and the column-vector of values of $\final(C)$,
listed in the same order, by $f$. Fix such numeration of nonterminals of the original
GF(2)-grammar. Now, we can indice both rows and columns of $n \times n$ matrices
by the nonterminals of the original GF(2)-grammar.

Let $I$ be an identity $n \times n$ matrix and $A$ be a $n \times n$ matrix, where
the cell on the intersection of $C$-th row and $E$-th column contains the sum 
$\Left(D)$ over all rules $C \to DE$ of the original grammar:	
\begin{equation}\label{abc_matrix_A} A_{C, E} := \sum\limits_{C \to DE} \Left(D)
\end{equation}
	
Similarly, let $B$ be $n \times n$ matrix with sum of $\Right(E)$ over all rules $C \to DE$
of the original grammar standing on the intersection of $C$-th row and $D$-th column
(it would make more sense to call this matrix $C$ rather than $B$, but we have already used the letter
$C$ for a different purpose):
\begin{equation}\label{abc_matrix_B} B_{C, E} := \sum\limits_{C \to DE} \Right(E)
\end{equation}

Then, System~\eqref{series_abc} can be stated in the following
compact matrix form: $x = f + (A + B)x$, or $(A + B + I)x = f$, which is the same.

As we have shown above, the column-vector of $\Center(C)$ values indeed is a solution
to such a system. If we somehow establish that this system has only one solution, and
the said solution can be expressed in relatively simple algebraic terms, we will get an expression for
for $\Dual(L(S_{a \to c})) = \Center(S)$.

This system has exactly one solution if and only if $\det (A + B + I) \neq 0$.
If $\det (A + B + I) \neq 0$, then, by Cramer's formula, each component of the solution,
$\Center(S)$, in particular, can be represented in the following form:
\begin{equation*} 
\dfrac{\det (A + B + I, \text{ but one of the columns was replaced by $f$})}{\det (A + B + I)}
\end{equation*}

Now, we still need to prove three things: that $\det (A + B + I)$ matches $\sum \A\C$,
that 
\begin{equation*}
\det (A + B + I, \text{ but one of the columns was replaced by $f$})
\end{equation*}
 matches
$\dfrac{\sum \A\B\C}{\poly(a, b) \poly(b, c)}$, independently of the replaced column,
and, finally, that $\det (A + B + I)$ is not zero. 

Each entry of $A + B + I$ matches $\A + \C$, because of Equations~\eqref{abc_matrix_A}
and \eqref{abc_matrix_B}. Indeed, each entry of $A$ match $\sum \A = \A$,
similarly, each component of $B$ matches $\C$, and entry of $I$ are zeroes and ones,
which match both $\A$ and $\C$.

The determinant of the matrix with all entries matching $\A + \C$, matches $\sum \A\C$.
We have already proven that during the proof of Lemma~\ref{Translation}. The proof
of $\det (A + B + I)$ being non-zero also comes from the same place verbatim.

Finally, determinant of $A + B + I$ with one column replaced by $f$ is something new, 
because entries of $f$ match rather complicated expression $\sum \left( \dfrac{\sum\A\B}{\poly(a, b)} \right) \cdot \left( \dfrac{\sum B\C}{\poly(b, c)} \right) = \dfrac{\sum \A\B\C}{\poly(a,b) \poly(b, c)}$. By using the formula
for determinant with $n!$ summands, we see that the determinant matches
\begin{equation*}
\sum (\A + \C) \ldots (\A + \C) \left( \dfrac{\sum \A\B\C}{\poly(a,b)\poly(b, c)} \right) (\A + \C)
\ldots (\A + \C). 
\end{equation*}
Here, in each summand, exactly one factor is complicated and the others
are very simple. By expanding all brackets, the determinant matches 
$\sum \dfrac{\sum \A\B\C}{\poly(a,b)\poly(b,c)}$. By taking the lowest common denominator
of all fractions in the sum, the determinant matches $\dfrac{\sum \A\B\C}{\poly(a,b)\poly(b,c)}$.

Now, $\Center(S) = L(S_{a \to c})$ is not exactly the language described by the new GF(2)-grammar, $L(S') = L(S_{a \to a}) \triangle L(S_{a \to b}) \triangle L(S_{a \to c})$ is. 
However, $\Dual(L(S')) = \Dual(L(S_{a \to a})) + \Dual(L(S_{a \to b})) + \Dual(L(S_{a \to c}))$.
Therefore, series $\Dual(K) = \Dual(L(S'))$ match $\A + \dfrac{\sum \A\B}{\poly(a,b)} + \dfrac{\sum \A\B\C}{\poly(a, b)\poly(b,c) \sum \A\C} = \dfrac{\sum \A\B\C}{\poly(a, b)\poly(b,c) \sum \A\C}$. The last equivalence 
holds, because the first two summands are simple and are ``absorbed'' by complicated third summand. 
\end{proof}

\subsection{Algebraic manipulations}

We will establish the following theorem:
\begin{thm}\label{Main_ab_oldc} Let $L$ be a subset of $a^* b^* c^*$ described
by a GF(2)-grammar. Then, the formal power series $\Dual(L)$ match algebraic 
expression $\dfrac{\sum \A\B\C}{\poly(a, b) \poly (b, c) \sum \A\C}$.
\end{thm}

By Lemma~\ref{Translation_abc}, it is enough to prove the following lemma:
\begin{lm} The algebraic expressions $\dfrac{\sum \A\B\C}{\poly(a, b) \poly (b, c) \sum \A\C}$ 
and $\dfrac{\sum \A\B\C}{\poly(a, b) \poly (b, c) \poly (a, c)}$ are equivalent. 
\end{lm}
\begin{proof} The proof is much simpler than the proof of Theorem~\ref{Main_ab_old}, because
we can use it now.

The second expression is not stronger than the first, because $\poly(a,c)$ is not stronger
than $\sum \A\C$.

On the other hand, we already know that the expression $\dfrac{\sum\A\C}{\sum\A\C}$
is not stronger than $\dfrac{\sum \A\C}{\poly(a,c)}$, because we needed that to prove
Theorem~\ref{Main_ab_old}. Therefore, the expression $\dfrac{\sum \A\B\C}{\poly(a, b) \poly (b, c) \sum \A\C}$ is not stronger than $\dfrac{\sum \A\B\C \cdot \sum \A\C}{\poly(a, b) \poly (b, c) \poly(a,c)}$. In the last expression, the factor $\sum \A\C$ can be ``absorbed'' into
$\sum \A\B\C$, giving us exactly the expression $\dfrac{\sum \A\B\C}{\poly(a, b) \poly (b, c) \poly(a,c)}$.
\end{proof}

\begin{rem} In a very similar way, with induction over the number $k$ of letters in the alphabet, we can prove
the following result: for $K \subseteq a_1^* a_2^* \ldots a_k^*$, the corresponding 
power series $\Dual(K)$ match the expression 
$\dfrac{\sum \prod\limits_{i=1}^n \A_i}{\prod\limits_{1 \leqslant i < j \leqslant n} \poly(a_i, a_j)}$, where $\A_i$ is the set of algebraic formal power series of variable $a_i$. 
\end{rem}

\end{appendices}

\end{document}